\documentclass[11pt]{article}
\usepackage{a4,amsmath,amsthm,amssymb,epsfig,amsthm,subfigure,wasysym}
\newtheorem{Theorem}{Theorem}
\newtheorem{Definition}{Definition}
\newtheorem{Lemma}{Lemma}
\newtheorem{Corollary}{Corollary}

\newcommand{\ith}[1]{${#1}^{\mbox{\footnotesize th}}$}

\begin{document}

\title{\textbf{Multidimensional Divide-and-Conquer\\and Weighted Digital Sums}
\thanks{HKUST authors' work was partially supported by HK RGC CRG 613105.}}

\author{
\normalsize Y. K. CHEUNG
\thanks{Computer Science Dept, Courant Institute of Mathematical Sciences, New York Universty.
Work done while at Dept of Mathematics, Hong Kong UST. \emph{ykcheung@cims.nyu.edu}}
\and 
\normalsize Philippe FLAJOLET
\thanks{INRIA Rocquencourt, F-78153, Le Chesnay (France). \emph{Philippe.Flajolet@inria.fr}}
\and
\normalsize Mordecai GOLIN
\thanks{Dept of Computer Science \& Engineering, Hong Kong UST. \emph{golin@cs.ust.hk}}
\and
\normalsize C. Y. James LEE
\thanks{Work done while at Dept of Mathematics, Hong Kong UST. \emph{mateddy@gmail.com}}}

\maketitle

\begin{abstract}\setlength{\parskip}{.1in}\setlength{\parindent}{0in}

\noindent This paper studies three types of functions arising separately in the analysis of algorithms that we analyze exactly using similar Mellin transform techniques.

The first is the solution to a Multidimensional Divide-and-Conquer (MDC) recurrence that arises when solving problems on points in $d$-dimensional space.

The second involves weighted digital sums. Write $n$ in its binary representation
$n=(b_i b_{i-1}\cdots b_1 b_0)_2$
and set
$S_M(n) = \sum_{t=0}^i t^{\overline{M}} b_t 2^t$.
We analyze the average
$TS_M(n) = \frac{1}{n}\sum_{j<n} S_M(j)$.

The third is a different variant of weighted digital sums. Write $n$ as
$n=2^{i_1} + 2^{i_2} + \cdots + 2^{i_k}$
with $i_1 > i_2 > \cdots > i_k\geq 0$
and set $W_M(n) = \sum_{t=1}^k t^M 2^{i_t}$.
We analyze the average
$TW_M(n) = \frac{1}{n}\sum_{j<n} W_M(j)$.

We show that both the MDC functions and $TS_M(n)$ (with $d=M+1$) have solutions of the form
$$\lambda_d n \lg^{d-1}n + \sum_{m=0}^{d-2}\left(n\lg^m n\right)A_{d,m}(\lg n) + c_d,$$
where $\lambda_d,c_d$ are constants and $A_{d,m}(u)$'s are periodic functions with period one (given by absolutely convergent Fourier series). 
We also show that $TW_M(n)$ has a solution of the form
$$n G_M(\lg n) + d_M \lg^M n + \sum_{d=0}^{M-1}\left(\lg^d n\right)G_{M,d}(\lg n),$$
where $d_M$ is a constant, $G_M(u)$ and $G_{M,d}(u)$'s are again periodic functions with period one (given by absolutely convergent Fourier series).
\end{abstract}

\section{Introduction}\label{sect:intro}

In this paper we use Mellin Transform techniques to analyze three types of functions arising separately in the analysis of algorithms: Multidimensional Divide-and-Conquer and two different types of weighted digital sums.

\par\noindent\underline{(A) Multidimensional Divide-and-Conquer:}\\
The Multidimensional Divide-and-Conquer (MDC) recurrence first appeared in the description of the running time of algorithms for finding maximal points in multidimensional space.
Previous analyses by Monier \cite{MONIE-1980} gave only first order asymptotic,
showing that the running time for the $d$-dimensional version of the problem is
($\lg n \equiv \log_2 n$)
$$T_d(n) = \lambda_d n \lg^{d-1} n + o(n \lg^{d-1} n)$$
for some constant $\lambda_d$.
We will extend the Mellin Transform techniques for solving divide-and-conquer problems originally developed in \cite{FlGo-1994}
(see \cite{GrHw-2005} for a review of more recent innovations)
to derive \emph{exact} solutions, which will be in the form of
\begin{equation}\label{eq:TSMn-closed-form}
T_d(n) = \lambda_d n \lg^{d-1}n + \sum_{m=0}^{d-2}\left(n \lg^m n\right)A_{d,m}(\lg n) + c_d,
\end{equation}
where $\lambda_d,c_d$ are constants and $A_{d,m}(u)$'s are periodic functions with period one given by absolutely convergent Fourier series.

\par\noindent\underline{(B) Weighted Digital Sums of the First Type:}\\
The second type of function we study is a generalization of weighted digital sums (WDS).
Start by representing  integer $n$ in binary as $n=(b_i b_{i-1}\cdots b_1 b_0)_2$. Define
$$S_1(n) := \sum_{t=0}^{i} t b_t 2^t,$$
i.e., \emph{weight} the \ith{t} digit by its location in the representation.
One can also view this as $2$ times
$\sum_{t=0}^{i} t b_t 2^{t-1},$,
which is analogous to the \emph{derivative} of the binary representation of $n$.
This sum arises naturally in the analysis of binomial queues where Brown \cite{BROWN-1978} gave upper and lower bounds
$$\lceil n\lg n-2n\rceil\leq S_1(n)\leq\lfloor n\lg n\rfloor.$$

Generalizing $S_1(n)$ allows the ``weights'' to be any polynomial of $t$.
Set $S_0(n) := n$ and $\forall M\geq 1$, define
\begin{equation}\label{eq:def-SMn}
S_M(n) := \sum_{t=0}^i t^{\overline{M}} b_t 2^t,
\end{equation}
where $t^{\overline{M}} := t(t+1)(t+2)\cdots(t+M-1)$ is the \ith{M} rising factorial of $t$.
The function $S_M(n)$ is not smooth (see Figure \ref{fig:SMn-second-asymptotic}).
We will instead analyze its \emph{average}
\begin{equation}\label{eq:def-TSMn}
TS_M(n) := \frac{1}{n}\sum_{j<n} S_M(j).
\end{equation}
We will show that, surprisingly, $TS_M(n)$ has an exact formula,
which is in exactly the same form (\ref{eq:TSMn-closed-form}) derived above for the MDC problem
(with $d=M+1$ and different constants).

\begin{figure}[t]
\vspace*{-.1in}
\centering%
\scalebox{0.33}{\includegraphics{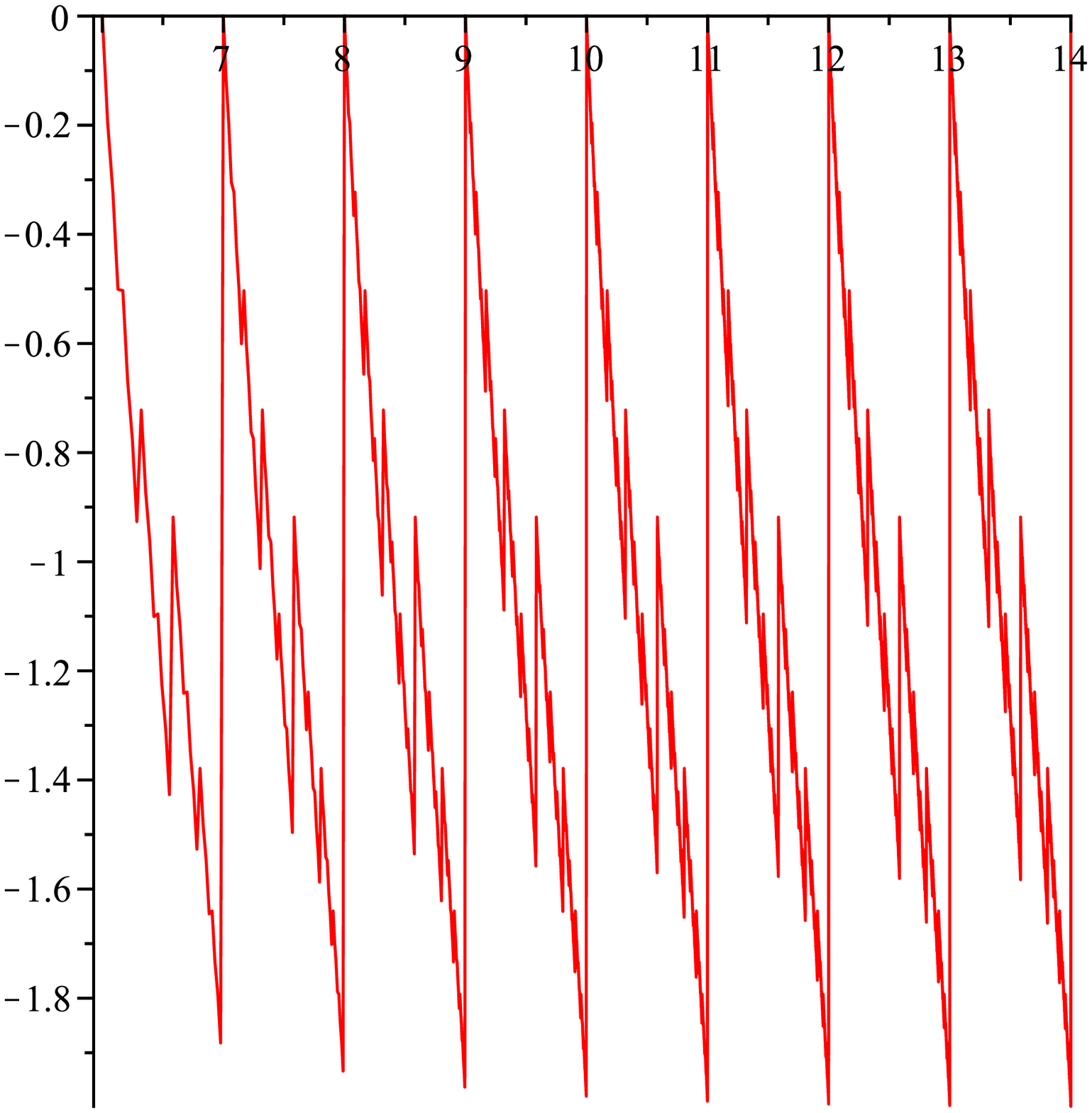}}
\scalebox{0.33}{\includegraphics{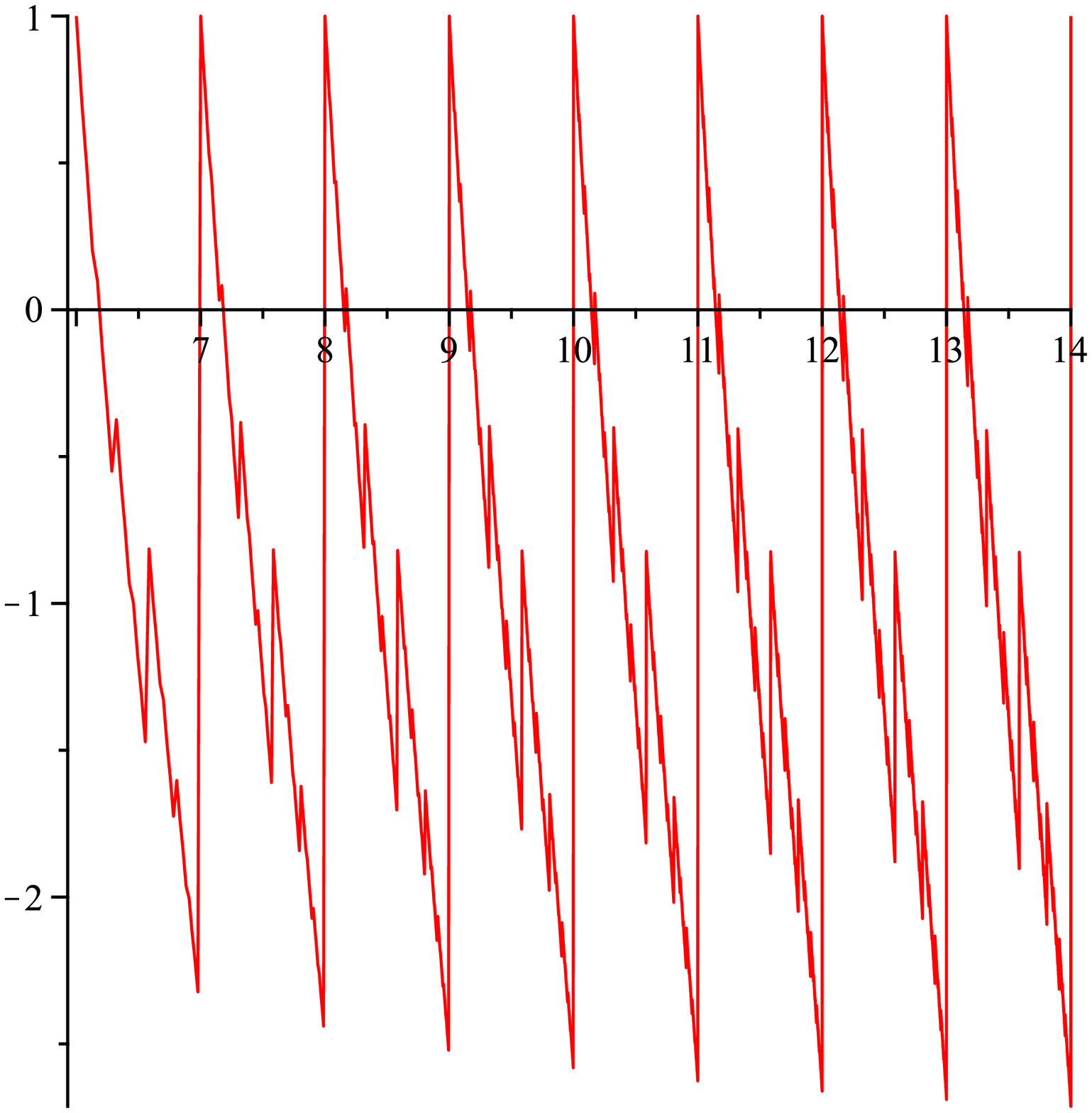}}
\caption{The graphs of $(S_M(n)-n\lg^M n)/(n\lg^{M-1}n)$ for $M=1$ (left) and $M=2$ (right) plotted against $\lg n$.
Although the functions appear periodic they possess ``large'' fluctuations.
These make direct analysis of $S_M(n)$ difficult, suggesting the analysis of its smoothed average instead.}
\label{fig:SMn-second-asymptotic}
\end{figure}

\par\noindent\underline{(C) Weighted Digital Sums of the Second Type:}\\
The third type of function we study is another WDS variant.
Its simplest form arises when analyzing the worst-case running time of bottom-up mergesort.
Assume\footnote{The actual worst-case time is $n_1+n_2-1$. But, any mergesort uses
exactly $n-1$ merges, so the running time derived with cost $n_1+n_2$ is exactly $n-1$
more than the real worst-case running time.}
that the worst-case running time to merge two sorted lists of sizes $n_1$ and $n_2$
into one sorted list is $n_1 + n_2$.

Define $C_w(n)$ to be the worst-case running time of bottom-up mergesort with $n$ elements.
Bottom-up mergesort essentially splits a list of $n$ items into two sublists,
sorts each recursively, and then merges them back together.
If $n$ is a power of $2$, then it splits the list into two even parts.
If $n$ is not a power of $2$ though,  i.e., $n=2^k+j$ with $1\leq j\leq 2^k-1$, then the
algorithm splits the items into one list of size $2^k$, one list of size $j$.
Thus it is known that $C_w(n)$ satisfies the recurrences:
\begin{eqnarray*}
C_w(2^k) & =&  k2^k.\\
C_w(2^k+j) & = & C_w(2^k) + C_w(j) + (2^k+j), \quad\mbox{ for $1\leq j\leq 2^k-1$}.
\end{eqnarray*}

Panny and Prodinger \cite{PaPr-1995} derived an exact solution for $C_w(n)$ containing a term $G(\log n)$,
where $G(x)$, defined by a Fourier series, is periodic with period one.
However, the Fourier series is only Ces\`{a}ro summable.
Furthermore, $G(x)$ is discontinuous at all dyadic points
(points of the form $x=n/2^m$, where $n$ is integer, $m$ is non-negative integer),
which are exactly the points of interest.

In this paper we will decompose $C_w(n)$ into two different types of WDS
and analyze the (smoothed) average of each part.
We will then generalize the functions found and analyze the generalizations.

Starting with the binary representation of $n$, ignore the $0$ bits
and write $n$ as the sum of descending powers of $2$, i.e.
$n=2^{i_1} + 2^{i_2} + \cdots + 2^{i_k}$ with $i_1 > i_2 > \cdots > i_k\geq 0$.
Iterating the above recurrence for $C_w(n)$ gives
$$C_w(n) = \sum_{t=1}^k i_t 2^{i_t} + \sum_{t=1}^k t 2^{i_t} - 2^{i_k} = S_1(n) + \sum_{t=1}^k t 2^{i_t} - 2^{i_k},$$
where $S_1(n)$ is the WDS of the first type defined previously.
This motivates the introduction of another variant of WDS:
\begin{equation}\label{eq:def-Wn}
W_1(n) := \sum_{t=1}^k t 2^{i_t}.
\end{equation}
As with the WDS of the first type, $W_1(n)$ is not smooth enough to be analyzed directly
(see Figure \ref{fig:WMn-first-asymptotic}),
so we instead study its average 
\begin{equation}\label{eq:def-TWn}
TW_1(n) := \frac{1}{n}\sum_{j<n} W_1(j).
\end{equation}

Similar to the WDS of the first type, this problem may be generalized by weighting
the powers of $2$ with polynomial weights\footnote{
For WDS of the first type we use weights of the form $t^{\overline{M}}$;
for WDS of the second type the weights are of the form $t^M$.
The difference is due to ease of analysis.
Both types of span the space of polynomials and hence our study allows \emph{any} polynomial weights.},
i.e. by defining
$W_0(n):=n$ and, $\forall M\geq 1$,
\begin{equation}\label{eq:def-WMn}
W_M(n) := \sum_{t=1}^k t^M 2^{i_t}
\end{equation}
and then introducing the average functions
\begin{equation}\label{eq:def-TWMn}
TW_M(n) := \frac{1}{n}\sum_{j<n} W_M(j).
\end{equation}

\begin{figure}[t]
\vspace*{-.1in}
\centering%
\scalebox{0.33}{\includegraphics{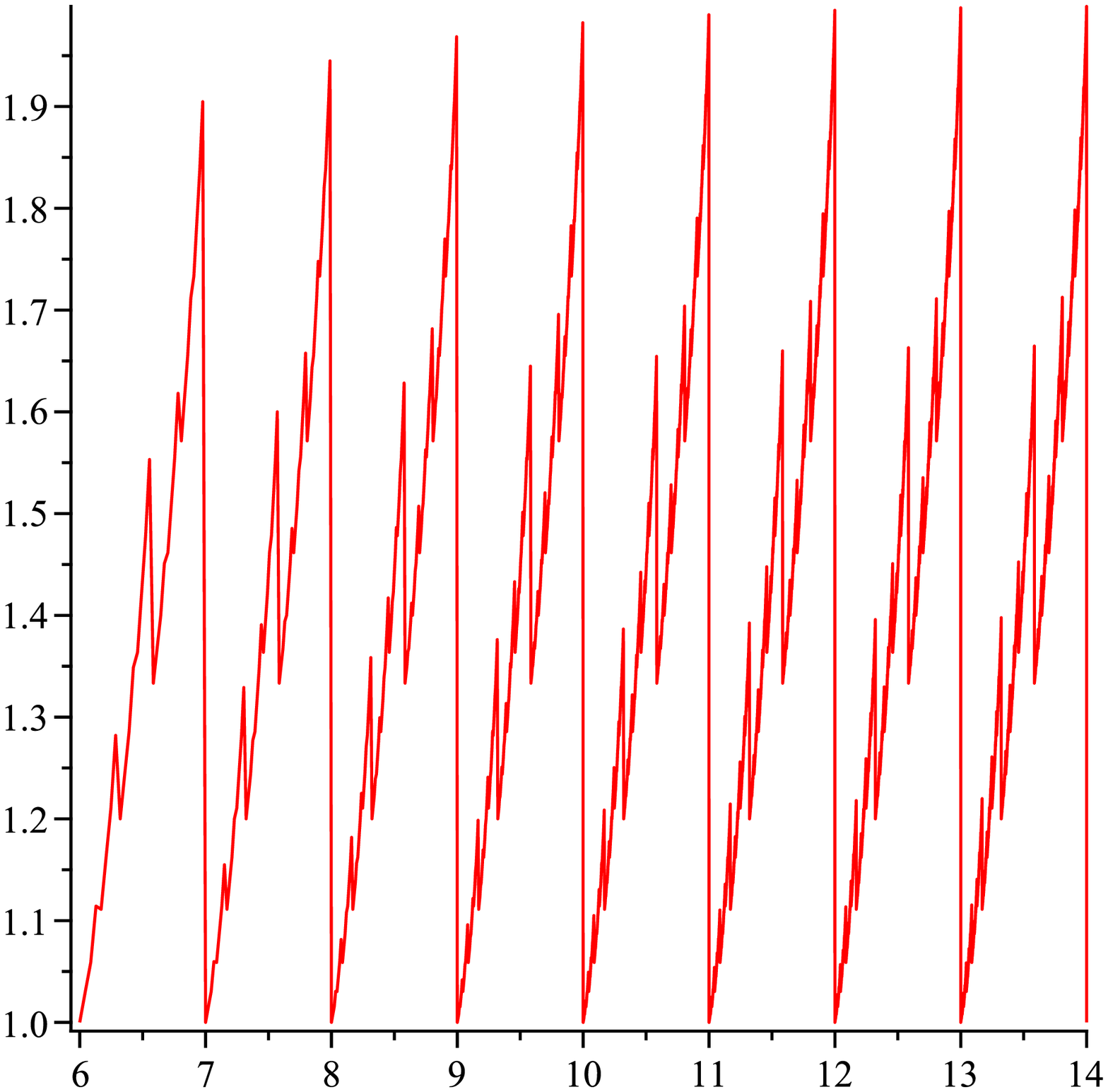}}
\scalebox{0.33}{\includegraphics{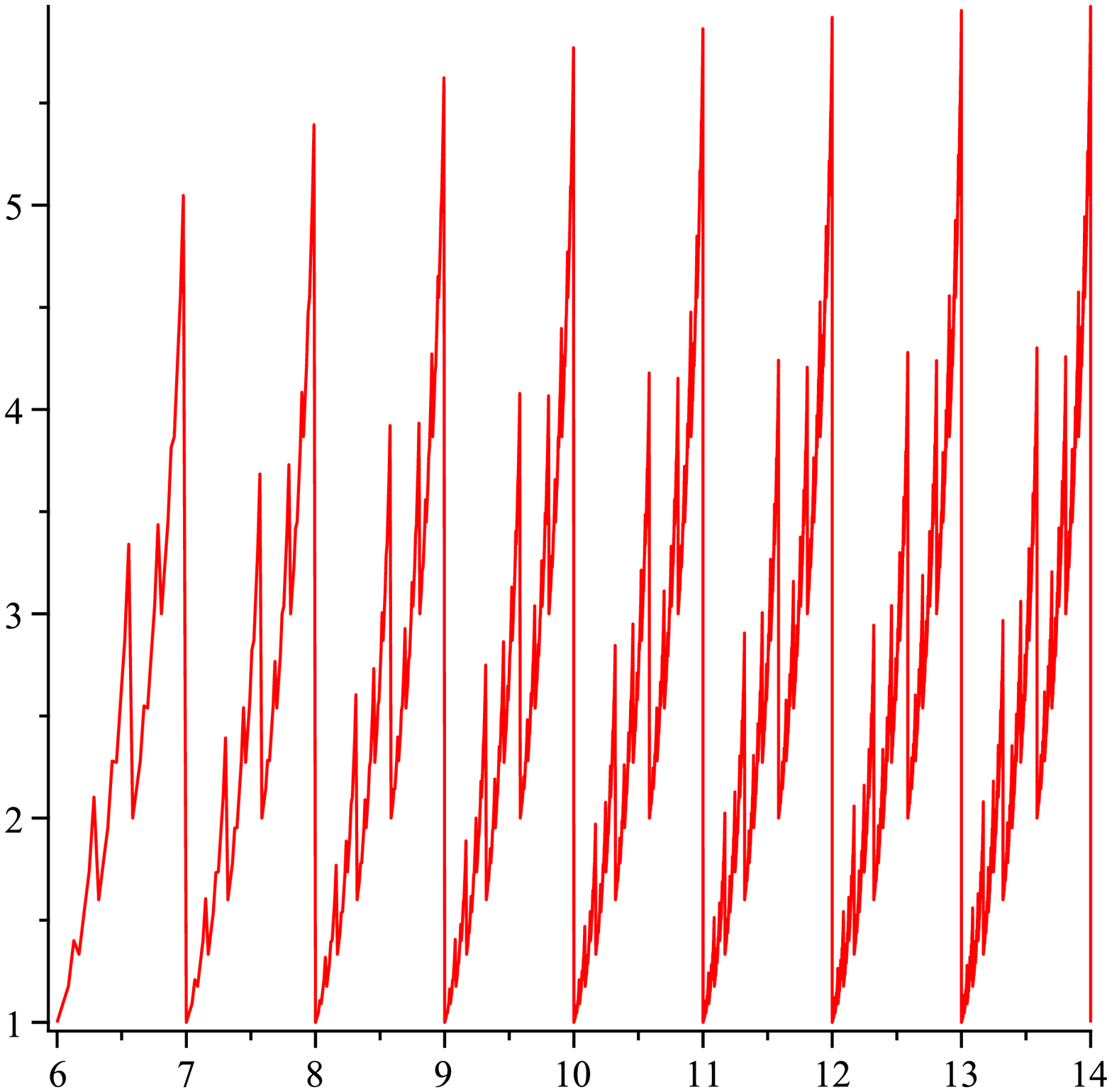}}
\caption{$W_M(n)/n$ plotted against $\lg n$ for $M=1$ (left) and $M=2$ (right).
Though these functions also appear periodic, they are not smooth and are with ``large'' fluctuations.
These make direct analysis of $W_M(n)$ hard.}
\label{fig:WMn-first-asymptotic}
\end{figure}

We will show that $TW_M(n)$ has an exact closed-form formula, which is in the form of
\begin{equation}\label{eq:TWMn-closed-form}
TW_M(n) = n G_M(\lg n) + d_M \lg^M n + \sum_{d=0}^{M-1} 
\left(\lg^d n\right)G_{M,d}(\lg n),
\end{equation}
where $d_M$ is a constant, $G_M(u)$ and $G_{M,d}(u)$'s are periodic functions with period one given by absolutely convergent Fourier series.

\medskip

Our approach to solving all three problems will be similar.
We first use the Mellin-Perron formula and problem-specific facts to reduce the analysis
to the calculation of an integral of the form
$$\int_{3-i\infty}^{3+i\infty} K(s) ds$$
for some problem specific kernel $K(s)$.
We then identify the singularities and residues of $K(s)$
and use the Cauchy residue theorem in the limit to evaluate the integral.

\begin{figure}[htp]
\centering%
\subfigure[$\left(f_n^3-\frac{1}{2}n\lg^2 n\right)/(n\lg n)$ vs.  $\lg n$]{
\scalebox{0.31}{\includegraphics{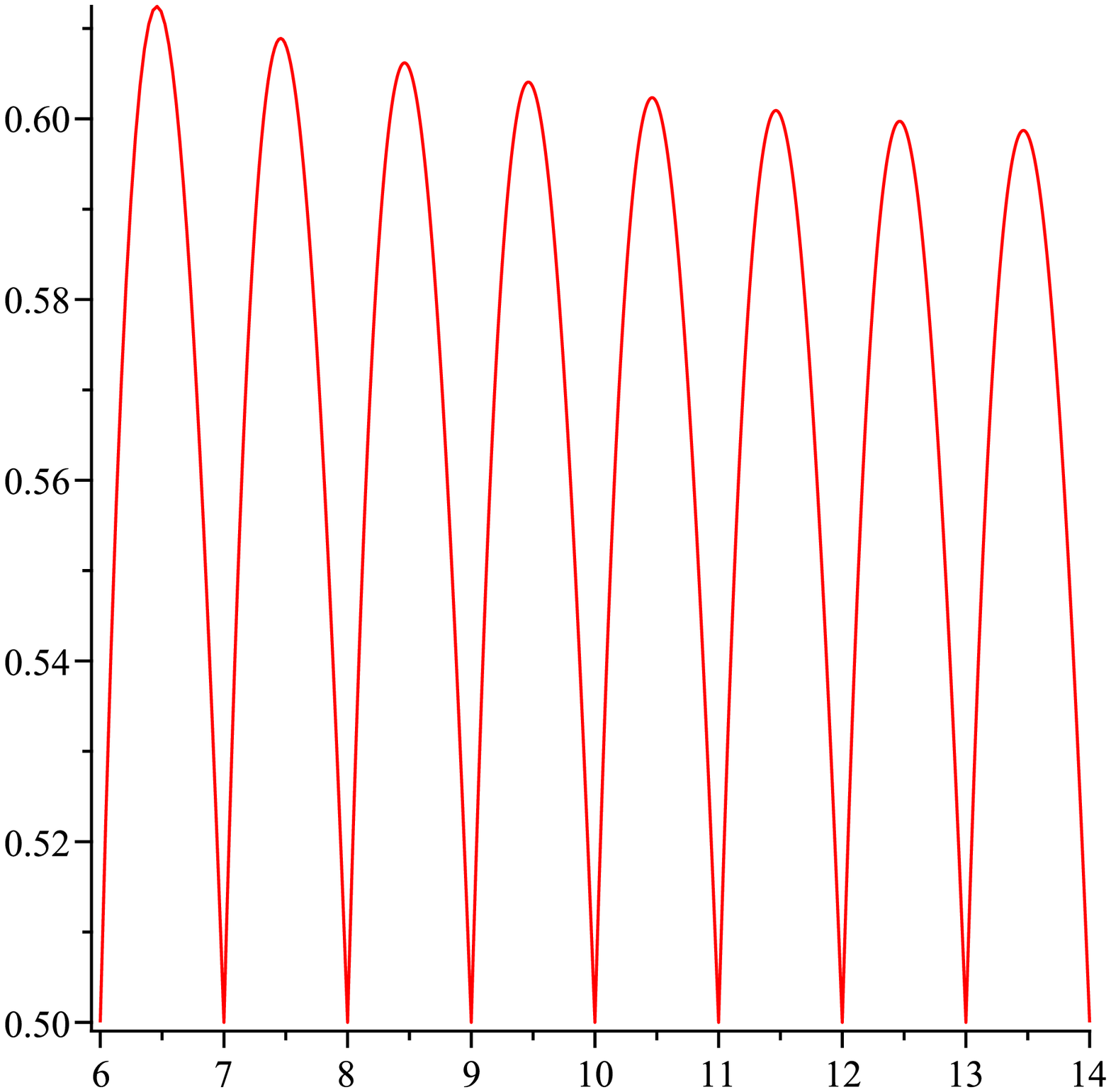}}
\label{fig:fn3_second_asym}}
\subfigure[$\left(f_n^4-\frac{1}{6}n\lg^3 n\right)/(n\lg^2 n)$ vs. $\lg n$]{
\scalebox{0.31}{\includegraphics{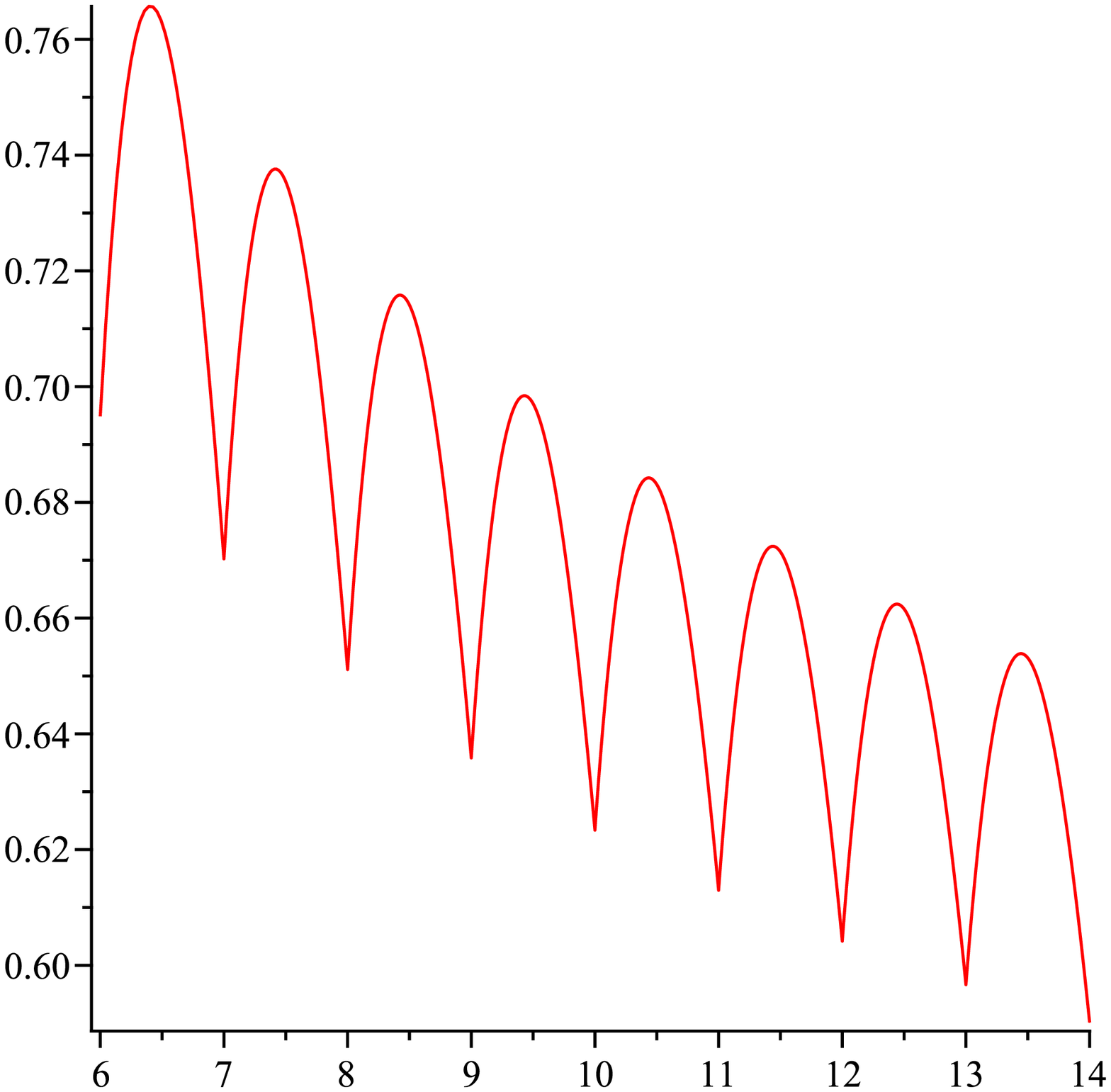}}
\label{fig:fn4_second_asym}}
\subfigure[$\left(TS_1(n)-\frac{1}{2}n\lg n\right)/n$ vs. $\lg n$]{
\scalebox{0.31}{\includegraphics{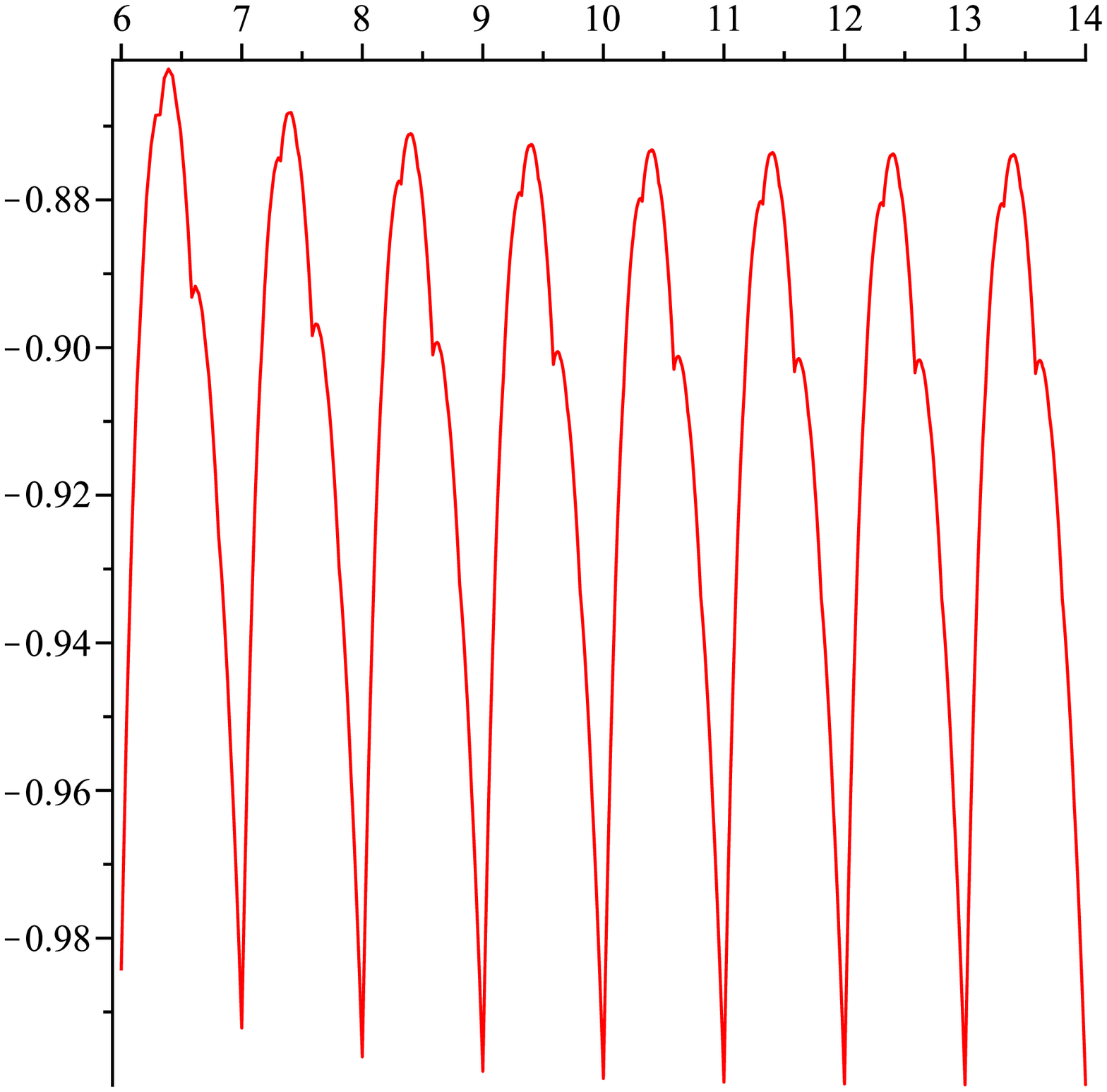}}
\label{fig:TS1n_second_asym}}
\subfigure[$\left(TS_2(n)-\frac{1}{2}n\lg^2 n\right)/(n\lg n)$ vs. $\lg n$]{
\scalebox{0.31}{\includegraphics{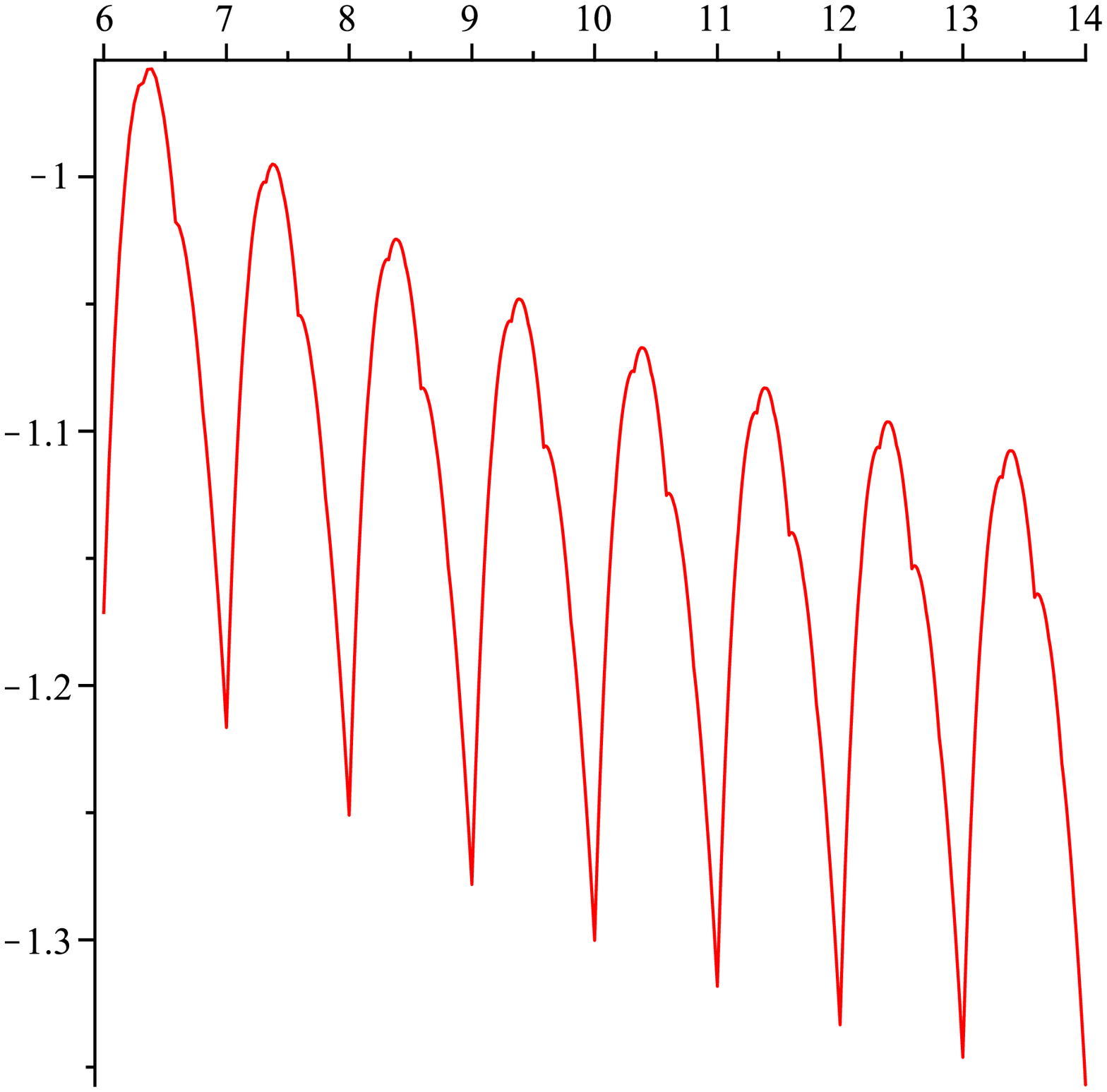}}
\label{fig:TS2n_second_asym}}
\subfigure[$TW_1(n)/n$ vs. $\lg n$]{
\scalebox{0.31}{\includegraphics{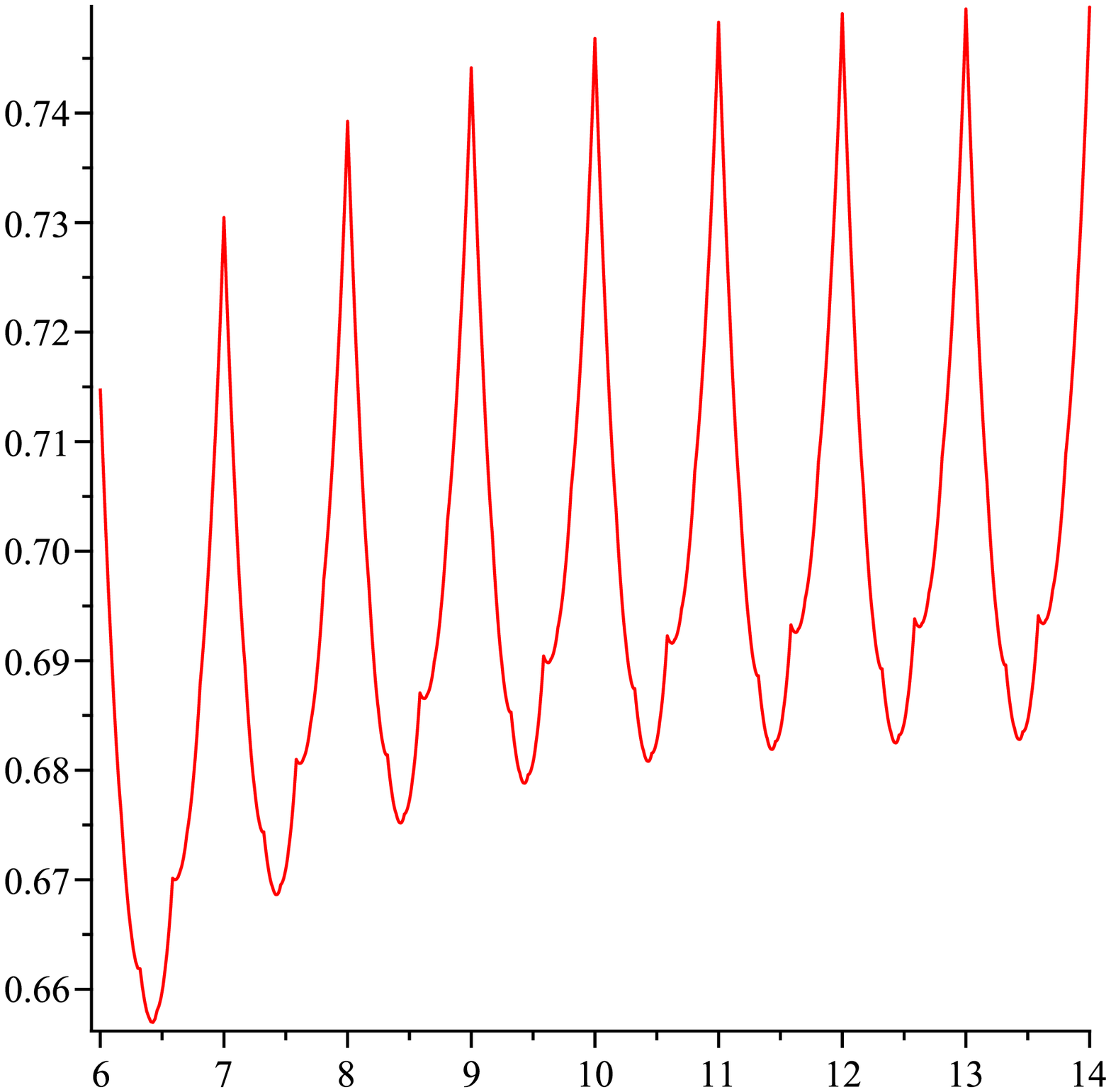}}
\label{fig:TW1n_first_asym}}
\subfigure[$TW_2(n)/n$ vs. $\lg n$]{
\scalebox{0.31}{\includegraphics{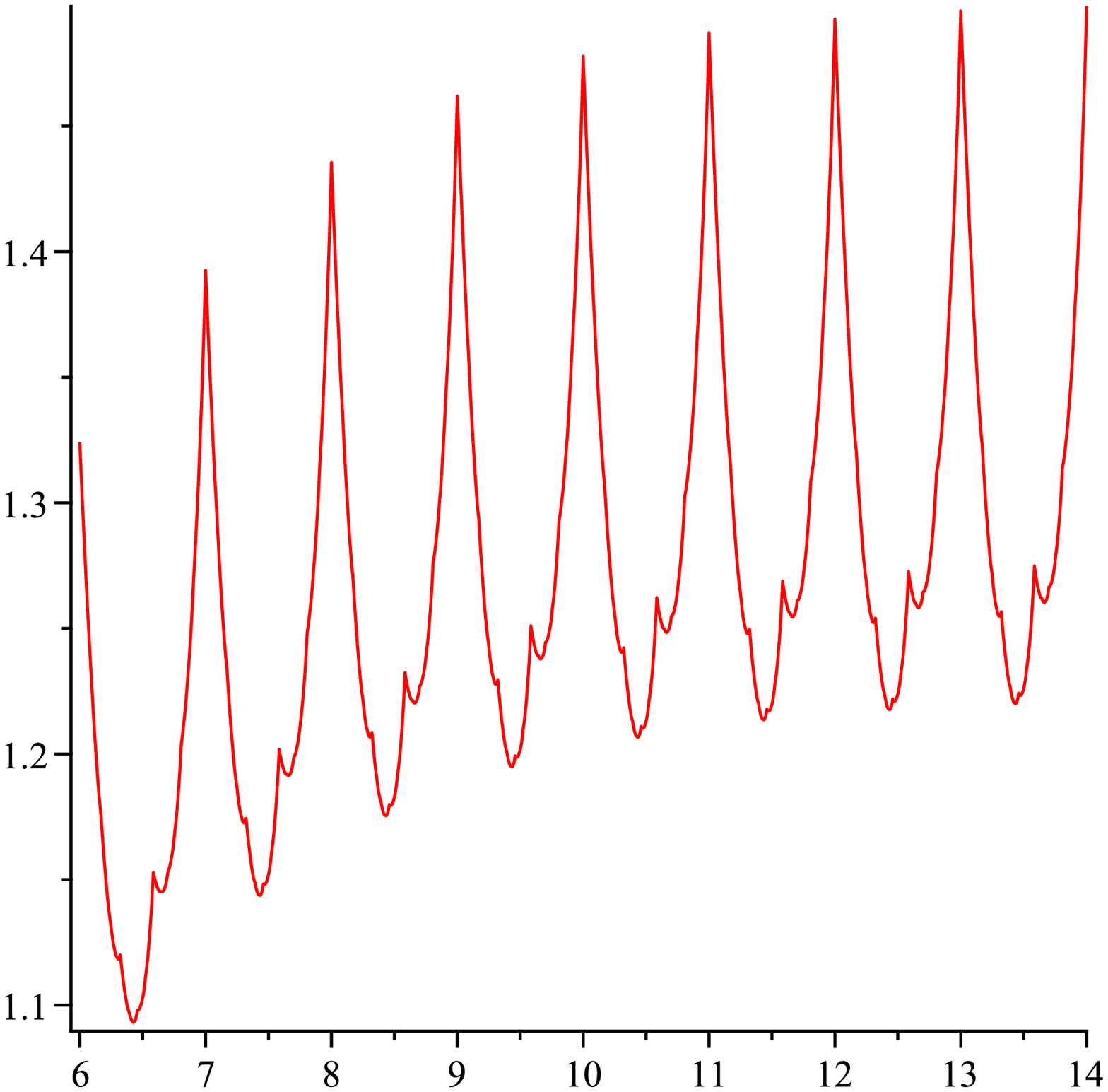}}
\label{fig:TW2n_first_asym}}
\caption{These figures illustrate the periodic nature of the second order asymptotics of
$f_n^k$ (when $k=3,4$),
$TS_M(n)$ (when $M=1,2$)
and the first order asymptotics of $TW_M(n)$ (when $M=1,2$).}
\end{figure}

\section{Background}\label{sect:back}

\subsection{The Mellin-Perron Fromula}\label{subsect:back-MPF}

The main tools used in this paper are Dirichlet generating functions and the Mellin-Perron formula.
For more background see, \cite[pp.13-23]{RIEDEL-MT-1996}, \cite{FGKPT-1994} and \cite[pp.762-767]{FlSe-2008}.

\begin{Theorem}[The Mellin-Perron formula]\label{thm:MP}
Let $\{\lambda_j\}$, $j=1,2,\ldots$ be a sequence and $c>0$ lie in the half-plane of absolute convergence of $\sum_{j=1}^{\infty} \lambda_j j^{-s}$. Then for any $m\geq 1$, 
\begin{equation}\label{eq:MPF-general-m}
\frac{1}{m!} \sum_{j<n} \lambda_j \left(1-\frac{j}{n}\right)^m = \frac{1}{2\pi i} \int_{c-i\infty}^{c+i\infty} \left(\sum_{j=1}^\infty \frac{\lambda_j}{j^s} \right) \frac{n^s ds}{s(s+1)(s+2)\cdots (s+m)}.
\end{equation}
\end{Theorem}

In particular, when $m=1$ and $m=2$,
\begin{eqnarray}
\frac{1}{n}\sum_{j<n} \lambda_j (n-j) & = & \frac{1}{2\pi i} \int_{c-i\infty}^{c+i\infty} \left(\sum_{j=1}^{\infty} \frac{\lambda_j}{j^s}\right)\frac{n^s ds}{s(s+1)},\label{eq:MPF-m=1}\\
\frac{1}{2n^2}\sum_{j<n} \lambda_j (n-j)^2 & = & \frac{1}{2\pi i} \int_{c-i\infty}^{c+i\infty} \left( \sum_{j=1}^{\infty} \frac{\lambda_j}{j^s} \right)\frac{n^s ds}{s(s+1)(s+2)}.\label{eq:MPF-m=2}
\end{eqnarray}

We will analyze the MDC functions and the two types of WDS
by rewriting them to summations in the form of the left hand side of (\ref{eq:MPF-m=1}).
WDS of the second type will also need a summation as in the left hand side of (\ref{eq:MPF-m=2}).
The Mellin-Perron formula will then enable us to evaluate the associated line integrals instead.
The line integrals will be evaluated \emph{exactly} via the Cauchy residue theorem
by considering integrations over some special contours.

In the right hand side of (\ref{eq:MPF-general-m}),
$\sum_{j=1}^\infty\lambda_j j^{-s}$,
the \emph{Dirichlet generating function} (DGF) of $\{\lambda_j\}$,
is the only factor depending upon $\{\lambda_j\}$.
The Cauchy residue theorem relates the value of the line integral to the residues at the poles of the kernel in the line integral,
thus understanding the locations and associated residues of the DGF's singularities
will be essential to evaluating the line integral.

Define the \emph{backward difference function} $\nabla A$ by $\nabla A(j) = A(j)-A(j-1)$ for any function $A$.
The following lemma will be needed later in the analysis of WDS.

\begin{Lemma}\label{thm:key}
Let $A$ be a function with $A(0)=0$ and
$$TA(n) = \frac{1}{n}\sum_{j<n} A(j).$$
Then 
$$TA(n) = \frac{1}{2\pi i}\int_{c-i\infty}^{c+i\infty}\left(\sum_{j=1}^\infty \frac{\nabla A(j)}{j^s}\right)\frac{n^s ds}{s(s+1)},$$
where $c>0$ lies in the half-plane of absolute convergence of $\sum_{j=1}^\infty \nabla A(j)j^{-s}$.
\end{Lemma}
\begin{proof}
Note that
$TA(n) = \frac{1}{n}\sum_{j<n} (n-j)\nabla A(j)$
and then apply (\ref{eq:MPF-m=1}).
\end{proof}

A similar lemma, previously proven by Flajolet and Golin \cite{FlGo-1994},
will be needed to analyze the MDC functions.
For any sequence $\{s_n\}$, define its \emph{double difference sequence}
$\{\Delta\nabla s_n\}$ by $\Delta\nabla s_n := s_{n+1} - 2 s_n + s_{n-1}$ for all $n$.

\begin{Lemma}\label{lem:solution-f-integral}
Consider the recurrence
\begin{equation}\label{eq:general-MDC-recurrence}
f_n = f_{\lfloor n/2 \rfloor} + f_{\lceil n/2 \rceil} + e_n
\end{equation}
with boundary conditions $e_0=e_1=0$ and $f_1=0$. Then
\begin{equation}\label{eq:FG}
D_{f_k}(s) = \sum_{j=1}^\infty\frac{\Delta\nabla f_j}{j^s} 
= \frac{1}{1-2^{-s}}\sum_{j=1}^{\infty} \frac{\Delta\nabla e_j^k}{j^s} 
\end{equation}
and 
$$f_n = \frac{n}{2\pi i}\int_{c-i\infty}^{c+i\infty} D_{f_k}(s) \frac{n^s ds}{s(s+1)},$$
where $c$ lies in the half-plane of absolute convergence of
$\sum_{j=1}^{\infty} \Delta\nabla e_j j^{-s}$.
\end{Lemma}

\subsection{Useful Facts Involving the Riemann-Zeta Function}\label{subsect:back-useful_identities}

The Riemann-Zeta function is defined by $\zeta(s) := \sum_{i>0} i^{-s}$ when $Re(s)>1$.
Since it will appear in the integral kernels in the analyses of the WDS, we list some basic facts concerning the Riemann-Zeta function \cite{TITCH-1986,WaWh-1963} that we will need.

First, $\zeta(s)$ can be analytically continued to be analytic in the whole complex plane
with the exception of a simple pole at $s=1$ with residue $1$.

Next, in \cite{FGKPT-1994}, Flajolet et. al. proved the identity
\begin{equation}\label{eq:left=0-m=1}
\frac{1}{2\pi i}\int_{-1/4-i\infty}^{-1/4+i\infty}\zeta(s)\frac{n^s ds}{s(s+1)} = 0.
\end{equation}
By mimicking their proof, we prove the similar formula
(for completeness the proof is provided in Appendix \ref{app:left-zero}):
\begin{equation}\label{eq:left=0-m=2}
\frac{1}{2\pi i}\int_{-5/4-i\infty}^{-5/4+i\infty} \zeta(s) \frac{n^s ds}{s(s+1)(s+2)} = 0.
\end{equation}

When integrating $\zeta(s)$ the following asymptotic bounds \cite{WaWh-1963} will be useful:

\begin{Lemma}\label{lem:zeta-bound}
If $s=\sigma+it$, where $\sigma,t\in\mathbb{R}$, the Riemann-Zeta function satisfies the bound
\begin{equation}\label{eq:zeta-bound-WW}
\zeta(s) = O(|t|^{\tau(\sigma)} \log |t|)
\end{equation}
where
\begin{equation}\label{eq:zeta-bound-regions-WW}
\tau(\sigma) = 
\begin{cases}
\frac{1}{2} - \sigma & \text{ for } \sigma \leq 0 \\
\frac{1}{2} 			& \text{ for } 0 \leq \sigma \leq \frac{1}{2} \\
1-\sigma 				& \text{ for } \frac{1}{2} \leq \sigma \leq 1 \\
0 							& \text{ for } 1 \leq \sigma.
\end{cases}
\end{equation}
\end{Lemma}

\subsection{Useful Formulae Involving Some DGFs}\label{subsect:back-DGF}

To understand the locations and associated residues of the integral kernels,
we will need closed-form formulae of their associated DGFs. We start with some basic definitions.

\begin{Definition}\label{def:v-and-v2}
Express $n=(b_i b_{i-1}\cdots b_1 b_0)_2$ in its binary representation.\\
Set $v(n) := \sum_{t=0}^i b_t$ to be the number of ``1''s in the binary representation of $n$
 and $v_2(n)$ to be the number of trailing ``0''s in the binary representation of $n$.\\
For example,\\
if $n=44=(101100)_2$ then $v(n) = 3$ and $v_2(n)=2$;\\  
if $n=33=(100001)_2$ then $v(n) = 2$ and $v_2(n)=0$.
\end{Definition}

We can now introduce two useful DGFs.

\begin{Definition}\label{def:VM-ZM}
$\forall M\geq 0$, denote the DGFs of $v(n)^M$ and $(v(n)+v_2(n))^M$ by
$$V_M(s) := \sum_{j=1}^\infty\frac{v(j)^M}{j^s},
\quad\quad
Z_M(s) := \sum_{j=1}^\infty\frac{(v(j)+v_2(j))^M}{j^s}.$$
\end{Definition}

The analysis of these DGFs will require the following facts.

\begin{Lemma}\label{lem:property-v-and-v_2}
Let $n$ be a positive integer. Then
\begin{enumerate}
\item $v(2n)=v(n)$ and $v(2n+1)=v(n)+1$;
\item $v_2(2n)=v_2(n)+1$;
\item if $n$ is odd, $v_2(n)=0$;
\item $v(n)-v(n-1)=1-v_2(n)$.
\end{enumerate}
\end{Lemma}
\begin{proof}
If $n=(b_i b_{i-1}\cdots b_1 b_0)_2$ then $2n=(b_i b_{i-1}\cdots b_1 b_0,0)_2$ and $2n+1=(b_i b_{i-1}\cdots b_1 b_0,1)_2$, so $v(2n)=v(n)$ and $v(2n+1)=v(n)+1$.

If the binary representation of $n$ has $t$ trailing ``$0$''s, then the binary representation of $2n$ will have $t+1$ trailing ``$0$''s. This proves $v_2(2n)=v_2(n)+1$.

If $n$ is odd, the rightmost digit of the binary representation of $n$ must be $1$, i.e. there is no trailing ``0'' in the representation. Hence $v_2(n)=0$ for odd integer $n$.

Slightly rewriting $n$ as
$n=(b_i b_{i-1} b_{t+1},1,0,0\cdots 0,0)_2$
where $t=v_2(n)$ shows that
$n-1=(b_i b_{i-1} b_{t+1},0,1,1\cdots 1,1)_2$.
Therefore $\nabla v(n) = v(n) -v(n-1) = 1-v_2(n)$.
\end{proof}

From Lemma \ref{lem:property-v-and-v_2}, it is straightforward to prove the following lemma,
which includes formulae expressing some special DGFs in terms of $\zeta(s)$ and $V_M(s)$.
For completeness, we provide its proof in Appendix \ref{app:DGF}.

\begin{Lemma}\label{lem:DGF-closed}
For $M\geq 1$,
\begin{equation}\label{eq:DGF-odd-v}
\sum_{\mbox{\footnotesize{odd }}j}\frac{v(j)^M}{j^s} = \left(1-\frac{1}{2^s}\right)V_M(s).
\end{equation}
The following DGFs have closed-form formulae in terms of $\zeta(s)$:
\begin{equation}\label{eq:DGF-v2}
\sum_{j=1}^\infty\frac{v_2(j)}{j^s} = \frac{1}{2^s-1}\zeta(s),
\end{equation}
\begin{equation}\label{eq:DGF-nabla-v}
\sum_{j=1}^\infty\frac{\nabla v(j)}{j^s} = \frac{2^s-2}{2^s-1}\zeta(s).
\end{equation}
\end{Lemma}

\subsection{Absolute Convergence of Fourier Series}\label{subsect:back-abs_conv}

In all three problems, evaluating the line integrals of the kernels will reduce to
summations of residues at poles regularly spaced along a vertical line.
These summations will best be expressed as Fourier series.
To be useful, we will need to show that these Fourier series converge absolutely.
Our major tools will be the following two lemmas.

\begin{Lemma}\label{thm:bound-derivative}
Let $\epsilon>0$, $\sigma_0,t_0\in\mathbb{R}$, $t_0\geq 1+\epsilon$ and $f$ be a complex function. If
\begin{enumerate}
\item $f$ is analytic in $X= \{s=\sigma+it : \sigma\geq\sigma_0-\epsilon , |t|\geq t_0-\epsilon\}$  and 
\item $\exists A, B >0$ such that  $\forall \sigma+it \in X$,\   $|f(\sigma+it)| = O(|t|^A\log^B |t|)$,
\end{enumerate}
then, for every fixed integer $q>0$,
$$\forall \sigma\geq\sigma_0,\, \forall |t|\geq t_0, \quad 
|f^{(q)}(\sigma+it)| = O(|t|^A\log^B |t|).$$
\end{Lemma}
\begin{proof}
From the Cauchy integral formula,
for all $s=\sigma+it$ with $\sigma\geq\sigma_0$ and $|t|\geq t_0$,
$$f^{(q)}(s) = \frac{q!}{2\pi i}\varoint_\mathcal{C} \frac{f(z)}{(z-s)^{n+1}} dz,$$
where $\mathcal{C} = \{z : |z-s|=\epsilon\}$. Hence
\begin{eqnarray*}
|f^{(q)}(s)| &\leq& \frac{q!}{2\pi}\varoint_\mathcal{C} \left|\frac{f(z)}{(z-s)^{n+1}}\right| dz\\
&\leq& \frac{q!}{2\pi}\times (2\pi\epsilon)\times O(|t|^A\log^B |t|)\times \frac{1}{\epsilon^{n+1}}\\
&=& O(|t|^A\log^B |t|)
\end{eqnarray*}
for fixed $q$ and $\epsilon$.
\end{proof}

Before stating the next lemma,
we clarify that the statement ``$h(s)$ has a pole of order at most $N$ at $s=s_0$'',
allows the possibility that $h(s)$ is analytic at $s=s_0$ (and might even have a zero there).

\begin{Lemma}\label{thm:sum-residues}
Let 
$g(s) = L(s) f(s)\frac {n^s}{s(s+1)}$.
$\forall j \in \mathbb{Z}$, set $\theta_j = \sigma + \frac{2\pi j}{\ln 2}i$. If
\begin{enumerate}
\item $\forall j\in\mathbb{Z}\setminus\{0\}$, $f$ is analytic at $s=\theta_j$,
\item $\exists A<1,\, B\geq 0$, such that for all integers positive integers $q,$
 $|f^{(q)}(\theta_j)| = O(|j|^A\log^B |j|)$ (where the constant in the big $O$ may depend upon $q$)
\item $\forall j\in\mathbb{Z}$, $L(s)$ has a pole of order at most $n_1$ at $s=\theta_j$;\\
furthermore, the coefficients of the Laurent series of $L(s)$ are identical at each $s=\theta_j$,
\item $\frac{f(s)}{s(s+1)}$ has a pole of order at most $n_2$ at $s=\theta_0$,
\end{enumerate}
then the sum of residues at $s=\theta_j$ can be written in the form
\begin{equation}\label{eq:sum-residues}
\sum_{j\in\mathbb{Z}}\mbox{Res}(g(s),s=\theta_j) = \sum_{i=n_1}^{n_1+n_2-1} \lambda_i n^\sigma \lg^i n + \sum_{i=0}^{n_1-1} F_i(\lg n) n^\sigma \lg^i n,
\end{equation}
where the $\lambda_i$'s are constants and $F_i(u)$'s are periodic functions with period one given by their Fourier series $F_i(u) = \sum_{j\in\mathbb{Z}} a_{i,j}e^{2\pi iju}$.
Furthermore, all the Fourier series $F_i(u)$ are absolutely convergent.
\end{Lemma}

\begin{proof}
We first introduce a notation. When $r$ is clear from the context, 
$$[m\,|\,a_m , a_{m+1}, a_{m+2}, \cdots]$$
represents the Laurent series $\sum_{i=0}^\infty a_{m+i} (s-r)^{m+i}$.

We start by stating the Laurent series of each factor of $g(s)$ at $s=\theta_j$,
where $\forall j\in\mathbb{Z}\setminus\{0\}$:
\begin{eqnarray*}
L(s) &=& \left[-n_1 \,\Big|\, l_0, l_1, l_2, l_3, \cdots\right]\\
f(s) &=& \left[0 \,\Big|\, f(\theta_j), f'(\theta_j), \frac{1}{2!}f''(\theta_j), \frac{1}{3!}f^{(3)}(\theta_j),\cdots\right]\\
n^s &=& n^\sigma e^{2\pi ij\lg n}\left[0 \,\Big|\, 1, \ln n, \frac{1}{2!}\ln^2 n, \frac{1}{3!}\ln^3 n, \cdots\right]\\
\frac{1}{s} &=& \left[0 \,\Big|\, \frac{1}{\theta_j}, -\frac{1}{(\theta_j)^2}, \frac{1}{(\theta_j)^3},\cdots\right]\\
\frac{1}{s+1} &=& \left[0 \,\Big|\, \frac{1}{\theta_j+1}, -\frac{1}{(\theta_j+1)^2}, \frac{1}{(\theta_j+1)^3},\cdots\right]
\end{eqnarray*}
The residue of $g(s)$ at $s=\theta_j$ is obtained by multiplying all these series together
and extracting the coefficient of the term $(s-\theta_j)^{-1}$.
The residue will therefore be the sum of terms, each term of the form
$$l_{x_1}\times \frac{1}{x_2!}f^{(x_2)}(\theta_j)\times n^\sigma e^{2\pi ij\lg n}\frac{\ln^{x_3} 2}{x_3!}\lg^{x_3} n\times \frac{(-1)^{x_4}}{(\theta_j)^{x_4+1}}\times \frac{(-1)^{x_5}}{(\theta_j+1)^{x_5+1}},$$
where $\forall i, x_i\geq 0$ and $\sum_{i=1}^5 x_i = n_1-1$.

Hence the sum of these residues, when sorted according to the variable $x_3$, is
\begin{equation}\label{eq:residue-theta_j}
\sum_{x_3=0}^{n_1-1} n^\sigma\left(\lg^{x_3}n\right)\frac{\ln^{x_3}2}{x_3!}\left( \sum_{j\in\mathbb{Z}\setminus\{0\}} J_j(x_1,x_2,x_4,x_5,n_1-1-x_3) e^{2\pi ij\lg n}\right),
\end{equation}
where
$$J_j(x_1,x_2,x_4,x_5,r) = \sum_{x_1+x_2+x_4+x_5=r} \left(l_{x_1}\times \frac{1}{x_2!}f^{(x_2)}(\theta_j)\times\frac{(-1)^{x_4}}{(\theta_j)^{x_4+1}}\times \frac{(-1)^{x_5}}{(\theta_j+1)^{x_5+1}}\right).$$

The relevant Laurent series of $g(s)$ at $s=\theta_0$ are:
\begin{eqnarray*}
L(s) &=& \left[-n_1 \,\Big|\, l_0, l_1, l_2, l_3, \cdots\right]\\
n^s &=& n^\sigma\left[0 \,\Big|\, 1, \ln n, \frac{1}{2!}\ln^2 n, \frac{1}{3!}\ln^3 n, \cdots\right]\\
\frac{f(s)}{s(s+1)} &=& \left[-n_2 \,\Big|\, h_0, h_1, h_2, h_3, \cdots\right]
\end{eqnarray*}
By multiplying all these series together
and extracting the coefficient of the term $(s-\theta_0)^{-1}$,
the residue at $s=\theta_0$ is found to be of the form
\begin{equation}\label{eq:residue-theta_0}
\sum_{i=n_1}^{n_1+n_2-1} \lambda_i n^\sigma\lg^i n + \sum_{i=0}^{n_1-1} \lambda_i n^\sigma\lg^i n.
\end{equation}
The second summation in (\ref{eq:residue-theta_0}) combines with (\ref{eq:residue-theta_j})
to give the second summation in (\ref{eq:sum-residues})
and the first summation in (\ref{eq:residue-theta_0}) gives the first summation in (\ref{eq:sum-residues}).

We now prove the absolute convergence of the Fourier series. Take $M = \max_{0\leq i\leq n_1-1} |l_i|$.
Note that for $q\leq n_1-1$, $|f^{(q)}(\theta_j)| = O(|j|^A\log^B |j|)$. Thus
\begin{eqnarray*}
&& \left| J_j(x_1,x_2,x_4,x_5,n_1-1-x_3) \right|\\
&\leq&
\binom{n_1+2-x_3}{3}\left|l_{x_1}\times \frac{1}{x_2!}f^{(x_2)}(\theta_j) \times \frac{(-1)^{x_4}}{(\theta_j)^{x_4+1}}\times \frac{(-1)^{x_5}}{(\theta_j+1)^{x_5+1}}\right|\\
&\leq&
\binom{n_1+2-x_3}{3}\times M\times O(|j|^A\log^B |j|)\times O\left(\frac{1}{|j|^2}\right)\\
&=& O\left(\frac{\log^B |j|}{|j|^{2-A}}\right).
\end{eqnarray*}
Since $(2-A)>1$, the Fourier series is absolutely convergent.
\end{proof}

To conclude, we note that as we only upper bound the order of poles but do not know their exact order,
$\lambda_i$ may be zero and the $F_i(u)$'s may be constant functions, or even zero functions.

\section{Multidimensional Divide-and-Conquer}\label{sect:MDC}

\subsection{Background of Multidimensional Divide-and-Conquer}

\emph{Multidimensional Divide-and-Conquer} (MDC) was first introduced by Bentley and Shamos \cite{BeSh-1976,BENTL-1980}
in the context of solving multidimensional computational geometry problems.
The generic idea is to solve a problem on $n$ $d$-dimensional points by
(i) first splitting the points into two almost equal subsets and solving the problem seperately
on each subset, then
(ii) taking all $n$ points, projecting them down to $(d-1)$ dimensional space and solving the problem on the projected set, and finally
(iii) constructing a solution to the complete problem by intelligently combining the solutions to the 3 previously solved ones.
The recursion bottoms out when the dimension $d=2$, in which case a straightforward solution is given, or when $n=1$, which has a trivial solution.

The methodology can be applied to give good solutions for many problems,
including the Empirical Cumulative Distribution Function (ECDF) problem,
maxima, range searching, closest pair, and the all nearest neighbour problem.

Of particular interest to us is the \emph{all-points ECDF problem in $\mathbb{R}^k$} (ECDF-$k$).
For two points $x=(x_1,x_2,\cdots ,x_k)$, $y=(y_1,y_2,\cdots ,y_k)\in\mathbb{R}^k$,
we say $x$ \emph{dominates} $y$ if $x_i\geq y_i$ for all $1\leq i\leq k$.
Given a set $S$ of $n$ points in $\mathbb{R}^k$,
the \emph{rank} of a point $x$ is the number of points in $S$ dominated by $x$.
The ECDF-$k$ problem is to compute the rank of each point in $S$.

When $k=2$, a slight modification of bottom-up mergesort will solve ECDF-$2$ in $\Theta(n\log n)$ time.
Monier \cite{MONIE-1980} proposed an MDC algorithm for solving ECDF-$k$ for larger $k$,
based on the description of Bentley \cite{BENTL-1980}.
Monier analyzed the worst-case running time of this algorithm, $T(n,k)$,
described by the following recurrence:
\begin{equation}\label{eq:T_MDC_def}
T(n,k) = 
	\left\{
		\begin{array}{ll}
T\left(\left\lfloor\frac{n}{2}\right\rfloor,k\right) +T\left(\left\lceil\frac{n}{2}\right\rceil,k\right)+T(n,k-1)+n & \mbox{if $n>1, k >2$},\\
			1 & \mbox{if $n=1, k>2$},\\
			n \lg n & \mbox{if $n\geq 1, k=2$}.
		\end{array}
	\right.
\end{equation}

By translation into a combinatorial path-counting problem he derived the first order asymptotic of $T(n,k)$.
More specifically, he showed that, for fixed $k$,
$$T(n,k) = \frac{1}{(k-1)!} n \lg^{k-1} n + \Theta (n \lg^{k-2} n).$$

We will derive exact solutions for the ECDF-$k$ running time using Lemma \ref{lem:solution-f-integral} from \cite{FlGo-1994}.
To do so, we will have to slightly modify the case $k=2$ to have a more precise initial condition.
In what follows we will denote $T(n,k)$ by $f_n^k$.  The recurrences corresponding to (\ref{eq:T_MDC_def}) will be:
\begin{equation}\label{eq:ECDFk-recurrence_1}
f_n^k=
\begin{cases}
f_{\lfloor n/2\rfloor}^k+f_{\lceil n/2\rceil}^k+e^k_n, & n \geq 2\\
0, & n = 1
\end{cases}
\end{equation}
where
\begin{equation}\label{eq:ECDFk-recurrence_2}
e^k_n =
\begin{cases}
f_n^{k-1}+n-1, & k \geq 3 \\
n - 1, & k=2.
\end{cases}
\end{equation}

\subsection{Deriving the DGF}\label{subsect:MDC-DGF}

To use Lemma \ref{lem:solution-f-integral} 
first requires a better understanding of the DGF of $\Delta\nabla f_n^k$,
which we denote by  by $D_{f_k}(s)$.
Start by noting that, directly from the lemma,
$$D_{f_k}(s) = \frac{1}{1-2^{-s}}\sum_{j=1}^{\infty} \frac{\Delta\nabla e_j^k}{j^s}.$$
One can work out directly that $\Delta\nabla e_1^2 =1$ while,
for $j\geq 2$, $\Delta\nabla e_j^2 = 0$. Thus,
\begin{equation}\label{eq:Df2}
D_{f_2}(s) = \sum_{j=1}^{\infty} \frac{\Delta\nabla f_j^2}{j^s} = \frac{1}{1-2^{-s}}\sum_{j=1}^{\infty} \frac{\Delta\nabla e_j^2}{j^s} = \frac{1}{1-2^{-s}}.
\end{equation}

For $k\geq 3$, 
$$\Delta\nabla e_j^k =
\begin{cases}
\Delta\nabla f_j^{k-1}, &\mbox{for $j \geq 2$}\\
e_2^k = f^{k-1}_2 + 1 = \Delta\nabla f^{k-1}_1 + 1, &\mbox{for $j = 1$.}
\end{cases}$$
Hence 
\begin{eqnarray*}
D_{f_k}(s) &=& \frac{1}{1-2^{-s}} \sum_{j=1}^{\infty} \frac{\Delta\nabla e_j^k}{j^s}\\
&=& \frac{1}{1-2^{-s}} \left( \Delta\nabla f_1^{k-1} + 1 + \sum_{j=2}^{\infty} \frac{\Delta\nabla f_j^{k-1}}{j^s} \right)\\
&=& \frac{1}{1-2^{-s}} + \frac{D_{f_{k-1}}(s)}{1-2^{-s}}.
\end{eqnarray*}

Iterating the above recurrence with initial condition  (\ref{eq:Df2}) yields
$$D_{f_k}(s) = \frac{1}{1-2^{-s}} + \frac{1}{(1-2^{-s})^2} + \cdots + \frac{1}{(1-2^{-s})^{k-1}}.$$

From Lemma \ref{lem:solution-f-integral}, for $k>1$,
\begin{eqnarray}\label{eq:f_n^k-recurrence}
f_n^k &=& \frac{n}{2\pi i} \int_{3-i\infty}^{3+i\infty} \left( \sum_{d=1}^{k-1} \frac{1}{(1-2^{-s})^d} \right) \frac{n^s ds}{s(s+1)}\nonumber\\
&=& f_n^{k-1} + \frac{n}{2\pi i} \int_{3-i\infty}^{3+i\infty} \frac{1}{(1-2^{-s})^{k-1}}\frac{n^s ds}{s(s+1)}.
\end{eqnarray}

We note that Flajolet and Golin \cite{FlGo-1994} explicitly solve the $k=2$ boundary case:
$$f_n^2 = n \lg n + n A_0^2(\lg n) + 1$$
where, setting $\beta_j := \frac{2\pi j}{\ln 2}i$,
$$A_0^2(u) = \left(\frac{1}{2}-\frac{1}{\ln 2}\right) + \frac{1}{\ln 2}\sum_{j\in\mathbb{Z}\setminus\{0\}} \frac{1}{\beta_j (\beta_j+1)} e^{2\pi i j u}.$$

\subsection{Evaluation of Integrals}\label{subsect:MDC-evaluation}

We now evaluate the integral in (\ref{eq:f_n^k-recurrence}):
\begin{equation}\label{eq:def-I_fk(s)}
I_k := \frac{n}{2\pi i} \int_{3-i\infty}^{3+i\infty} \frac{1}{(1-2^{-s})^{k-1}}\frac{n^s ds}{s(s+1)}.
\end{equation}

Fix some real $R>0$ and consider the counterclockwise rectangular contour
$\Upsilon = \Upsilon_1 \bigcup \Upsilon_2 \bigcup \Upsilon_3 \bigcup \Upsilon_4$,
where (see Figure \ref{fig:contour_upsilon})
\begin{eqnarray}\label{eq:Upsilon}
&\Upsilon_1 = \{ 3+iy : -R \leq y \leq R \} \quad\quad\quad&\Upsilon_2 = \{ x+iR : -R \leq x \leq 3 \} \nonumber\\
&\Upsilon_3 = \{ -R+iy : -R \leq y \leq R \} \quad\quad&\Upsilon_4 = \{ x-iR : -R \leq x \leq 3 \} 
\end{eqnarray}

Denote the kernel of the integral in (\ref{eq:def-I_fk(s)}) by $K_k(s)$:
\begin{equation}\label{eq:kernel-fk}
K_k(s)=\frac{n^s}{(1-2^{-s})^{k-1}s(s+1)}.
\end{equation}

\begin{figure}[t]
\vspace*{-.1in}
\centering%
\scalebox{0.3}{\includegraphics{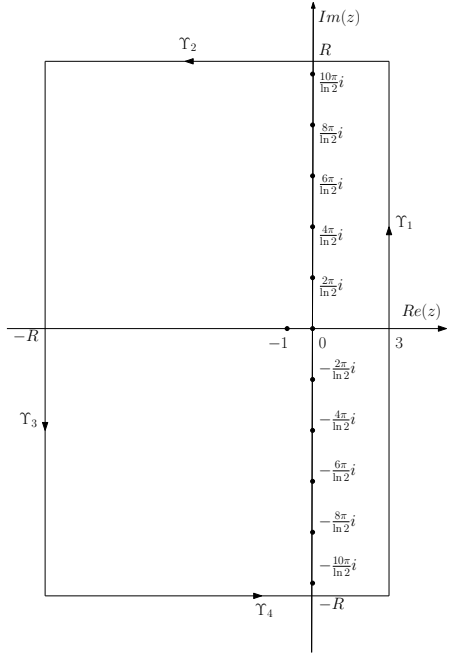}}\\
\caption{Contour $\Upsilon$ defined in (\ref{eq:Upsilon}).
The dots represent the poles of $K_k(s)$ inside $\Upsilon$.}
\label{fig:contour_upsilon}
\end{figure}

Note that $I_k = \lim_{R \rightarrow \infty} \frac{n}{2\pi i} \int_{\Upsilon_1} K_k(s) ds$.

We now show that, for $q=2,3,4$, $\lim_{R \rightarrow \infty} \int_{\Upsilon_q} K_k(s) ds =0$. Thus
$$I_k =\lim_{R\rightarrow\infty}\frac{n}{2\pi i}\int_\Upsilon K_k(s) ds.$$
By the Cauchy residue theorem, $I_k$ will be equal to $n$ times the sum of the residues at the poles inside $\Upsilon$ as $R\rightarrow\infty$.

The poles of $K_k(s)$ inside $\Upsilon$ are:
\begin{enumerate}
\item A pole of order $k$ at $s=0$;
\item Poles of order $(k-1)$ at $s=\beta_j=\frac{2\pi j}{\ln 2}i$, where $j\in \mathbb{Z}\setminus\{0\}$;
\item A simple pole at $s=-1$.
\end{enumerate}
To avoid poles of $K_k(s)$ on $\Upsilon$, we only consider values of $R=R_j := \frac{(2j+0.5) \pi}{\ln 2}$. 

Now consider the horizontal paths $q=2,4$. Then
\begin{eqnarray*}
\left| \int_{\Upsilon_q} K_k(s) ds \right| 
&\leq& \int_{-R_j\pm iR_j}^{3\pm iR_j} \left| K_k(s) \right| ds\\
&\leq&  \left(\max_{-R_j\leq \sigma\leq 3} \left| \frac{n^{\sigma}}{(1\pm 2^{-\sigma}i)^{k-1}} \right| \frac{1}{R_j(R_j+1)}\right) \int_{-R_j}^3 d\sigma = O(j^{-1}).
\end{eqnarray*}

For the leftmost path it is easy to see
$$\left| \int_{\Upsilon_3} K_k(s) ds \right| = O\left( \frac{1}{R_j (2^{k-1}n)^{R_j}} \right) = o(j^{-1}).$$

Hence $I_k$ is $n$ times the sum of the residues at the poles of $K_k(s)$ inside $\Upsilon$,
taking $R_j\rightarrow\infty$.

\begin{Theorem}\label{thm:MDC_equation}
For $k\geq 2$,
\begin{equation}\label{eq:fnk-general-exact}
f_n^k = \frac{1}{(k-1)!} n \lg^{k-1} n + \sum_{m=0}^{k-2} \left(n \lg^m n\right) A_m^k(\lg n) + c_k,
\end{equation}
where $A^k_m(u)$'s are periodic functions with period one,
which are given by absolutely convergent Fourier series
$$A_m^k(u) = \sum_{j\in\mathbb{Z}} a_{k,m,j} e^{2\pi iju}$$
whose coefficients $a_{k,m,j}$ can be determined explicitly.
In particular, the average value of $A_{k-2}^k(u)$ is
$$a_{k,k-2,0} = \frac{1}{(k-2)!}\left(\frac{k+1}{2}-\frac{1}{\ln 2}\right).$$
Furthermore, if $k$ is even, $c_k = 1$; if $k$ is odd, $c_k = 0$. 
\end{Theorem}
\begin{proof}
We proceed by induction on $k$. As previously mentioned, for $k=2$ this theorem was already proved by Flajolet and Golin \cite{FlGo-1994}.

Now assume that (\ref{eq:fnk-general-exact}) is true for $k=k_0-1$. 
The residue of $K_{k_0}(s)$ at $s=-1$ is
$$\mbox{Res}\left(K_{k_0}(s),s=-1\right)=\frac{(-1)^{k_0}}{n}.$$

We can now apply Lemma \ref{thm:sum-residues}, by taking $\sigma=0$, $L(s) = (1-2^{-s})^{-(k_0-1)}$ (its Laurent series coefficients at each $s=\beta_j$ are identical) and $f(s) = 1$. Since $f^{(q)}(s)\equiv 0$ when $q\geq 1$, we may take $A=B=0$. The order of poles of $L(s)$ at $s=\beta_j$ is $k_0-1$ and the order of pole of $\frac{f(s)}{s(s+1)}$ at $s=\beta_0=0$ is $1$.
 
The sum of residues at $s=\beta_j$, where $j\in\mathbb{Z}$, is given by
$$\lambda_{k_0-1} \lg^{k_0-1} n + \sum_{m=0}^{k_0-2}\left(\lg^m n\right) B_m^{k_0}(\lg n),$$
where $B_m^{k_0}(u)$'s are periodic functions with period one which are given by absolutely convergent Fourier series. $\lambda_{k_0-1}$ can be explicitly calculated to be $1/(k_0-1)!$.

Hence by (\ref{eq:f_n^k-recurrence}),
\begin{eqnarray*}
&& f_n^{k_0}\\
&=& f_n^{k_0-1} + n\left[ \frac{\lg^{k_0-1} n}{(k_0-1)!} + \sum_{m=0}^{k_0-2}\left(\lg^m n\right) B_m^{k_0}(\lg n) + \frac{(-1)^{k_0}}{n} \right]\\
&=& \frac{n \lg^{k_0-2} n}{(k_0-2)!} + \sum_{m=0}^{k_0-3} \left(n \lg^m n\right) A_m^{k_0-1}(\lg n) + c_{k_0-1}\\
&&\qquad + \frac{n\lg^{k_0-1} n}{(k_0-1)!} + \sum_{m=0}^{k_0-2}\left(n \lg^m n\right) B_m^{k_0}(\lg n) + (-1)^{k_0}\\
&=& \frac{n\lg^{k_0-1} n}{(k_0-1)!} + \left(\frac{1}{(k_0-2)!} + B_{k_0-2}^{k_0}(\lg n)\right)n\lg^{k_0-2} n \\
&&\qquad + \sum_{m=0}^{k_0-3}\left(A_m^{k_0-1}(\lg n) + B_m^{k_0}(\lg n)\right)n\lg^m n + c_{k_0-1} + (-1)^{k_0}.
\end{eqnarray*}

Letting $A_{k_0-2}^{k_0}(u) := \frac{1}{(k_0-2)!} + B_{k_0-2}^{k_0}(u)$ and $A_m^{k_0}(u) := A_m^{k_0-1}(u) + B_m^{k_0}(u)$ for $m=0,1,\cdots,k_0-2$ proves  (\ref{eq:fnk-general-exact}).
The average value of $A_{k-2}^k(u)$ is found by expressing all the Laurent series (in the proof of Lemma \ref{thm:sum-residues}) explicitly.

Finally, since $c_k =c_{k_0-1} + (-1)^{k_0}$,
$c_k$ alternates between being even and odd with $c_2=1$.
\end{proof}

\section{Weighted Digital Sums of the First Type}\label{sect:WDS1}
We now analyze $TS_M(n) = \sum_{j < n} S_M(j)$
as defined in (\ref{eq:def-SMn}) and (\ref{eq:def-TSMn}).
By Lemma \ref{thm:key}, this reduces to evaluating
\begin{equation}\label{eq:TSMn-integral-raw-dgf}
TS_M(n) = \frac{1}{2\pi i} \int_{c-i\infty}^{c+i\infty} \left( \sum_{j=1}^{\infty} \frac{\nabla S_M(j)}{j^s} \right)\frac{n^s ds}{s(s+1)}.
\end{equation}

\subsection{Deriving the DGF}\label{subsect:WDS1-DGF}

We start by deriving a closed form for 
\begin{equation}\label{eq:def-AMs}
A_M(s) := \sum_{j=1}^{\infty} \frac{\nabla S_M(j)}{j^s}.
\end{equation}

Recall that $S_M(n) = \sum_{t=0}^i t^{\overline{M}} b_t 2^t$.
Observe that if $n=(b_i b_{i-1}\cdots b_1 b_0)_2$, then 
$$2n=(b_i b_{i-1}\cdots b_1 b_0,0)_2\quad\mbox{and}\quad 2n+1=(b_i b_{i-1}\cdots b_1 b_0,1)_2.$$
In particular, when $M\geq 1$,
the weight $t^{\overline{M}}$ for the rightmost digit $(t=0)$ is always zero, so
\begin{equation}\label{eq:SM(2n+1)-recurrence}
S_M(2n+1) = S_M(2n).
\end{equation}

Next, observe that
\begin{equation}\label{eq:S1(2n)-recurrence}
S_1(2n) = \sum_{t=0}^{i} (t+1)b_t 2^{t+1} = 2\sum_{t=0}^{i} t b_t 2^t + 2 \sum_{t=0}^{i} b_t 2^t = 2 S_1(n) + 2n
\end{equation}
and for $M\geq 2$,
\begin{eqnarray}
S_M(2n)	&=& \sum_{t=0}^{i} (t+1)^{\overline{M}} b_t 2^{t+1} \nonumber\\
			&=& 2 \sum_{t=0}^i t^{\overline{M}} b_t 2^t + M \sum_{t=0}^{i} (t+1)^{\overline{M-1}} b_t 2^{t+1}\nonumber\\
&=& 2 S_M(n) + M S_{M-1}(2n). \label{eq:SM(2n)-recurrence}
\end{eqnarray}

These facts lead to:
\begin{Lemma}\label{lem:dgf-nabla-SM-exact}
\begin{equation}
\label{eq:AM_def}
A_M(s)=  M! \frac{2^{(M-1)(s-1)}}{(2^{s-1}-1)^M} \zeta(s).
\end{equation}
\end{Lemma}

\begin{proof}
The proof is by induction on $M$. When $M=1$, by (\ref{eq:SM(2n+1)-recurrence}) and (\ref{eq:S1(2n)-recurrence}), we get
$$\nabla S_1(2n) = 2\nabla S_1(n) + 2\quad\mbox{and}\quad \nabla S_1(2n+1) = 0.$$

Hence
$$A_1(s) = \sum_{j=1}^{\infty} \frac{\nabla S_1(j)}{j^s} = \sum_{l=1}^{\infty} \frac{\nabla S_1(2l)}{(2l)^s} = \sum_{l=1}^{\infty} \frac{2\nabla S_1(l)+2}{(2l)^s} = \frac{1}{2^{s-1}}(A_1(s)+\zeta(s)).$$
Then  $A_1(s) = (2^{s-1}-1)^{-1}\zeta(s)$ and the lemma is proved for $M=1$.

Now assume the lemma is true for $M<k$. Iterating (\ref{eq:SM(2n)-recurrence}) gives 
$$S_k(2n) = 2S_k(n) + 2 \left(\sum_{i=1}^{k-1} k^{\underline{i}} S_{k-i}(n)\right) + 2k!n,$$
where $k^{\underline{i}} = k(k-1)\cdots (k-i+1)$ is the \ith{i} falling factorial of $k$.

Appling  (\ref{eq:SM(2n+1)-recurrence}) gives
\begin{eqnarray*}
\nabla S_k(2n) &=& 2 \nabla S_k(n) + 2\left(\sum_{i=1}^{k-1} k^{\underline{i}} \nabla S_{k-i}(n)\right) + 2k!,\\
\nabla S_k(2n+1) &=& 0.
\end{eqnarray*}

Substituting the above two formulae into (\ref{eq:def-AMs}) yields
\begin{eqnarray*}
A_k(s) &=& \sum_{j=1}^\infty\frac{\nabla S_k(j)}{j^s} = \sum_{l=1}^\infty\frac{\nabla S_k(2l)}{(2l)^s}\\
&=& \sum_{l=1}^\infty\left(\frac{2 \nabla S_k(l) + 2\left(\sum_{i=1}^{k-1} k^{\underline{i}} \nabla S_{k-i}(l)\right) + 2k!}{(2l)^s}\right)\\
&=& \frac{1}{2^{s-1}}A_k(s) + \frac{1}{2^{s-1}}\sum_{i=1}^{k-1}k^{\underline{i}} A_{k-i}(s) + \frac{k!}{2^{s-1}}\zeta(s)\\
&=& \frac{1}{2^{s-1}}A_k(s) + \frac{1}{2^{s-1}}\left[\sum_{i=1}^{k-1}k^{\underline{i}}  (k-i)!\frac{2^{(k-i-1)(s-1)}}{(2^{s-1}-1)^{k-i}}\zeta(s) + k!\zeta(s)\right]\\
&=& \frac{1}{2^{s-1}}A_k(s) + \frac{k!\zeta(s)}{2^{s-1}}\left(1 + \sum_{i=1}^{k-1} \frac{2^{(k-i-1)(s-1)}}{(2^{s-1}-1)^{k-i}}\right)\\
&=& \frac{1}{2^{s-1}}A_k(s) + \frac{k!\zeta(s)}{2^{s-1}}\left(1 + \frac{1}{2^{s-1}} \sum_{i=1}^{k-1}\left(\frac{2^{s-1}}{2^{s-1}-1}\right)^{k-i}\right)\\
&=& \frac{1}{2^{s-1}}A_k(s) + \frac{k!\zeta(s)}{2^{s-1}}\left(\frac{2^{s-1}}{2^{s-1}-1}\right)^{k-1}
\end{eqnarray*}
and hence
$$A_k(s) = k! \frac{2^{(k-1)(s-1)}}{(2^{s-1}-1)^k} \zeta(s).$$
\end{proof}

\subsection{Evaluation of the Integral}\label{subsect:WDS1-evaluation}

Substituting the result of Lemma \ref{lem:dgf-nabla-SM-exact} into the integral of
 (\ref{eq:TSMn-integral-raw-dgf}) gives
\begin{equation}\label{eq:TSMn-integral-exact-dgf}
TS_M(n) = \frac{M!}{2\pi i} \int_{3-i\infty}^{3+i\infty} \frac{2^{(M-1)(s-1)}}{(2^{s-1}-1)^M} \zeta(s)\frac{n^s ds}{s(s+1)}.
\end{equation}

Fix some real $R>0$ and consider the counterclockwise rectangular contour
$\Gamma = \Gamma_1\cup\Gamma_2\cup\Gamma_3\cup\Gamma_4$, where (see Figure \ref{fig:contour_gamma1})
\begin{eqnarray}\label{eq:Gamma}
&\Gamma_1 = \{ 3+iy : -R \leq y \leq R \},
\quad\quad\quad&
\Gamma_2 = \{ x+iR : -1/4 \leq x \leq 3 \}, \nonumber\\
&\Gamma_3 = \{ -1/4+iy : -R \leq y \leq R \},\quad
&\Gamma_4 = \{ x-iR : -1/4 \leq x \leq 3 \}
\end{eqnarray}

\begin{figure}[t]
\vspace*{-.1in}
\centering%
\scalebox{0.38}{\includegraphics{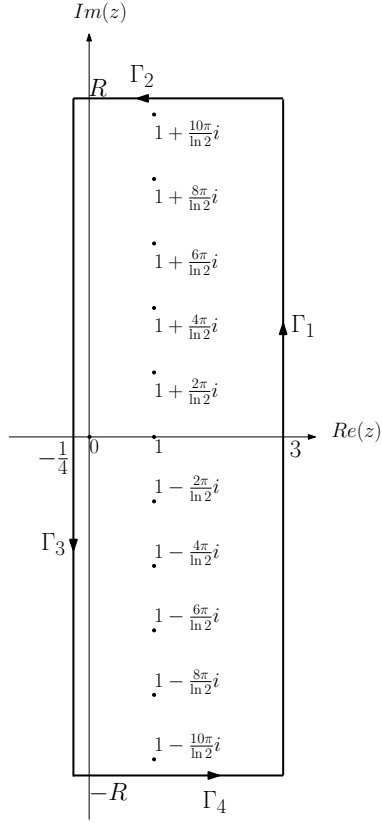}}
\caption[The Contour $\Gamma$ for evaluating $TS_M(n)$]{Contour $\Gamma$ from  (\ref{eq:Gamma}). The dots represent the poles of $K_M(s)$ inside $\Gamma$.}
\label{fig:contour_gamma1}
\end{figure}

Denote the kernel of the integral in (\ref{eq:TSMn-integral-exact-dgf}) by $K_M(s)$:
\begin{equation}\label{eq:kernel-TM}
K_M(s) = \frac{2^{(M-1)(s-1)}\zeta(s)n^s}{(2^{s-1}-1)^Ms(s+1)}.
\end{equation}
Note that $TS_M(n) = \lim_{R \rightarrow \infty} \frac {M!} {2\pi i} \int_{\Gamma_1} K_M(s) ds$. As in the MDC case, we now  show that $\lim_{R\rightarrow\infty}\int_{\Gamma_q} K_M(s) ds = 0$ for $q=2,3,4$. Thus 
$$TS_M(n) = \lim_{R\rightarrow\infty}\frac{M!}{2\pi i}\int_\Gamma K_M(s) ds.$$
Hence by the Cauchy residue theorem, $TS_M(n)$ will be equal to $M!$ times the sum of the residues at the poles inside $\Gamma$ as $R\rightarrow\infty$.

We know that $\zeta(s)$ has a simple pole at $s=1$. The poles of $K_M(s)$ inside $\Gamma$ are:
\begin{enumerate}
\item A pole of order $(M+1)$ at $s=1$;
\item Poles of order $M$ at $s=\alpha_j := 1 + \frac{2\pi j}{\ln 2}i$, where $j\in\mathbb{Z}\setminus\{0\}$;
\item A simple pole at $s=0$.
\end{enumerate}

To avoid  poles of $K_M(s)$ on $\Gamma$, we again only consider values of 
$R= R_j = \frac{(2j+0.5) \pi}{\ln 2}$.

To show that $\int_{\Gamma_q} K_M(s) ds=0$ for $q=2,3,4$ as $R\rightarrow\infty$, we
need the  following two lemmas.

\begin{Lemma}\label{thm:bound_top_bottom_zeta_general}
Consider integral 
$$I(R) = \int_{-a+iR}^{3+iR} f(s) \zeta(s) n^s ds,$$
where $0<a\leq\frac{5}{4}$. 
Furthermore, suppose that for $s=\sigma+it$ with $-a\leq\sigma\leq 3$, $|f(s)|=O(|t|^{-2})$. Then, both as
$R\rightarrow\infty$ and $R\rightarrow -\infty$, $I(R)\rightarrow 0$.
\end{Lemma}

\begin{proof}
Along the path of the integral, $\Re(s)\geq-5/4$. Lemma \ref{lem:zeta-bound} gives the bound
$$|\zeta(s)| = O(|R|^{7/4} \log |R|).$$

Together with the given fact that $|f(s)|=O\left(|t|^{-2}\right)$,
\begin{eqnarray*}
\left| \int_{-a+iR}^{3+iR} f(s) \zeta(s) n^s ds\right| & \leq & \int_{-a+iR}^{3+iR} \left| f(s) \zeta(s) n^s \right| ds \\
& \leq & \int_{-a+iR}^{3+iR} \left(O(|R|^{-2}) \times O(|R|^{7/4} \log |R|) \times n^3\right) ds \\
& \leq & (3+a) O(|R|^{-1/4} \log |R|) n^3 \\
& \rightarrow & 0
\end{eqnarray*}
as $R\rightarrow\infty$.
\end{proof}

\begin{Lemma}\label{thm:bound_left_zeta_general}
Suppose
$$g(s) = \sum_{j=0}^{\infty} g_j (K_j)^s$$
for some real sequence $\{g_j\}$ and positive integer sequence $\{K_j\}$.
If this series is uniformly convergent
for $ s \in \left\{-\frac{1}{4} + it \,:\, t \in \mathbb{R}\right\}$
then
$$\int_{-1/4-i\infty}^{-1/4+i\infty} g(s)\zeta(s)\frac{n^s ds}{s(s+1)} =0.$$
If
the series is uniformly convergent
for $ s \in \left\{-\frac{5}{4} + it \,:\, t \in \mathbb{R}\right\}$
then
$$\int_{-5/4-i\infty}^{-5/4+i\infty} g(s)\zeta(s)\frac{n^s ds}{s(s+1)(s+2)} = 0.$$
\end{Lemma}
\begin{proof}
For the first integral,  note that
\begin{eqnarray*}
\int_{-1/4-i\infty}^{-1/4+i\infty} g(s)\zeta(s)\frac{n^s ds}{s(s+1)}
& = & \int_{-1/4-i\infty}^{-1/4+i\infty} \left( \sum_{j=0}^{\infty} g_j (K_j)^s \right)\zeta(s)\frac{n^s ds}{s(s+1)}\\
& = & \sum_{j=0}^{\infty} \left(g_j \int_{-1/4-i\infty}^{-1/4+i\infty} \zeta(s)\frac{(K_jn)^s ds}{s(s+1)}\right)\\
& = & 0
\end{eqnarray*}
The first equality is the definition of $g(s),$  the second follows
from the uniform convergence of  the series and the last equality 
follows from (\ref{eq:left=0-m=1}).

The second integral is evaluated similarly, using (\ref{eq:left=0-m=2})
in place of (\ref{eq:left=0-m=1}).
\end{proof}

To evaluate the integrals along  $\Gamma_2$ and $\Gamma_4$, note that  $\left|2^{(M-1)(s-1)}(2^{s-1}-1)^{-M}\right|$ is bounded as $j\rightarrow\infty$ and $\left|\frac{1}{s(s+1)}\right|=O(j^{-2})$. Thus, by Lemma \ref{thm:bound_top_bottom_zeta_general},  as $R_j\rightarrow\infty$,
$$\int_{\Gamma_2} K_M(s) ds \rightarrow 0,
\quad
\int_{\Gamma_4} K_M(s) ds \rightarrow 0.$$

To evaluate the integral along $\Gamma_3$, note that $\sigma<0$ along $\Gamma_3$, so 
we may write
$$\frac{1}{2^{s-1}-1} = -1 - \left(\frac{1}{2}\right)2^s - \left(\frac{1}{4}\right)4^s - \left(\frac{1}{8}\right)8^s \cdots.$$

The series is both absolutely convergent and uniformly convergent on $-1/4 + (-\infty,\infty)i$, so we may write
 (see \cite[pp.74-75]{RUDIN-1976})
$$\frac{2^{(M-1)(s-1)}}{(2^{s-1}-1)^M} = \sum_{j=0}^{\infty} a_j (2^{M+j-1})^s$$
for some $\{a_j\}$, where this new series is again uniformly convergent on $-1/4 + (-\infty,\infty)i$. By Lemma \ref{thm:bound_left_zeta_general},
$$\lim_{R_j\rightarrow\infty} \int_{\Gamma_3} K_M(s) ds \rightarrow 0.$$

We have successfully shown that the integrals along $\Gamma_2,$ $\Gamma_3$ and $\Gamma_4$ vanish as $R_j\rightarrow\infty$, and hence $TS_M(n)$ is $M!$ times the sum of the residues at the poles of $K_M(s)$ inside $\Gamma$, after taking $R_j\rightarrow\infty$.

\begin{Theorem}\label{thm:TSn-general}
For $M\geq 1$, 
\begin{equation}\label{eq:TMn-general-exact}
TS_M(n) = \frac{1}{2} n \lg^M n + \sum_{d=0}^{M-1} \left(n\lg^d n\right)F_{M,d}(\lg n) + (-1)^{M+1} M!,
\end{equation}
where $F_{M,d}(u)$'s are periodic functions with period one, which are given by absolutely convergent Fourier series
$$F_{M,d}(u)=\sum_{j\in\mathbb{Z}} f_{M,d,j} e^{2\pi ij u}$$
whose coefficients $f_{M,d,j}$ can be determined explicitly. In particular, the average value of $F_{M,M-1}(u)$ is
$$f_{M,M-1,0} = \frac{M}{4\ln 2}[2\gamma_0-3+(M-2)\ln 2]\approx \frac{M^2}{4}-0.915648 M.$$
\end{Theorem}
\begin{proof}
As shown, $TS_M(n)$ is $M!$ times the sum of residues at the poles of $K_M(s)$ inside $\Gamma$ as $R\rightarrow\infty$. The residue of $K_M(s)$ at $s=0$ is
$$\mbox{Res}(K_M(s),s=0) = (-1)^{M+1}.$$

By Lemma \ref{lem:zeta-bound}, we have the bound $|\zeta(\sigma+it)|=O(|t|^\epsilon \log |t|)$ when $\sigma\geq 1-\epsilon$ for some sufficiently small $\epsilon$. By Lemma \ref{thm:bound-derivative}, $|\zeta^{(q)}(\alpha_j)| = O(|j|^\epsilon \log |j|)$ for any fixed positive integer $q$.

In Lemma \ref{thm:sum-residues}, take $\sigma=1$, $L(s) = 2^{(M-1)(s-1)}(2^{s-1}-1)^{-M}$ (its Laurent series coefficients at each $s=\alpha_j$ are identical) and $f(s)=\zeta(s)$. From last paragraph we can take $A=\epsilon$ and $B=1$. The order of poles of $L(s)$ at $s=\alpha_j$ is $M$, and the order of pole of $\frac{f(s)}{s(s+1)}$ at $s=\alpha_0=1$ is $1$.

The sum of residues at $s=\alpha_j$, where $j\in\mathbb{Z}$, is given by
$$\lambda_M n \lg^M n + \sum_{d=0}^{M-1} \left(n\lg^d n\right)\overline{F}_{M,d}(\lg n),$$
where $\overline{F}_{M,d}(u)$'s are periodic functions with period one which are given by absolutely convergent Fourier series. $\lambda_M$ can be explicitly calculated to be $1/(2M!)$.

Hence
\begin{eqnarray*}
TS_M(n) &=& M!\left[ \frac{n \lg^M n}{2M!} + \sum_{d=0}^{M-1} \left(n\lg^d n\right)\overline{F}_{M,d}(\lg n) + (-1)^{M+1} \right]\\
&=& \frac{1}{2}n\lg^M n + \sum_{d=0}^{M-1} \left(n\lg^d n\right)M!\overline{F}_{M,d}(\lg n) + (-1)^{M+1}M!.
\end{eqnarray*}
Letting $F_{M,d}(u) := M!\overline{F}_{M,d}(u)$ for $d=0,1,\cdots,M-1$, proves (\ref{eq:TMn-general-exact}). The average value of $F_{M,M-1}(u)$ is found by expressing all the Laurent series (in the proof of Lemma \ref{thm:sum-residues}) explicitly.
\end{proof}

\section{Weighted Digital Sums of the Second Type}\label{sect:WDS2_higher}

We now analyze $TW_1(n)= \sum_{j<n} W_1(j)$ as defined by (\ref{eq:def-Wn}) and (\ref{eq:def-TWn}). The analysis will be extended to $TW_M(n)$ for
$M>1$ in the next section.

The general methodology used to analyze $TW_1(n)$ is the same as in the previous sections; use Lemma \ref{thm:key} to rewrite
\begin{equation}\label{eq:TWn-integral-raw-dgf}
TW_1(n) = \frac{1}{2\pi i}\int_{c-i\infty}^{c+i\infty}\left(\sum_{j=1}^\infty \frac{\nabla W_1(j)}{j^s}\right) \frac{n^s ds}{s(s+1)}.
\end{equation}
The main difficulty that will be encountered  is that the DGF here  
will not be  ``nice'' enough to permit  integrating the kernel
directly. We will have to split the DGF into two parts, using the
$m=1$ case of (\ref{eq:MPF-general-m}) to evaluate the first part and the $m=2$ case
to evaluate the second part.

\subsection{Deriving the DGF}\label{subsect:WDS2-higher-DGF}

Set
\begin{equation}\label{eq:def-BMs}
B_M(s) := \sum_{j=1}^\infty\frac{\nabla W_M(j)}{j^s}.
\end{equation}
to be the DGF of $\nabla W_M(j)$. We start by deriving, for all $M\geq 1$,
a formula for $B_M(s)$ in terms of DGFs $V_M(s)$ and $Z_M(s)$ introduced in Definition \ref{def:VM-ZM}. We will then analyze the case $M=1$ in this section, and leave the cases $M\geq 2$ to the next section.

\begin{Lemma}\label{lem:dgf-nabla-WM-exact}
$$B_M(s)=\frac{2^s-1}{2^s-2}V_M(s)-\frac{1}{2^s-2}Z_M(s).$$
\end{Lemma}

\begin{proof}
Observe that if $n$ is expressed as  $n=2^{i_1} + 2^{i_2} + \cdots + 2^{i_k}$ with $i_1 > i_2 > \cdots > i_k\geq 0$, then
\begin{eqnarray*}
W_M(2n) &=& \sum_{t=1}^k t^M 2^{i_t+1} = 2\sum_{t=1}^k t^M 2^{i_t} = 2W_M(n),\\
W_M(2n+1) &=& \sum_{t=1}^k t^M 2^{i_t+1} + (k+1)^M = 2W_M(n)+(v(n)+1)^M.
\end{eqnarray*}
Recalling from Lemma \ref{lem:property-v-and-v_2} that $v(n)-v(n-1)=1-v_2(n)$ and $v(2n+1)=v(n)+1$ gives
$$\nabla W_M(2n)=2\nabla W_M(n)-(v(n)+v_2(n))^M\qquad\mbox{and}\qquad\nabla W_M(2n+1)=v(2n+1)^M.$$

Then, (\ref{eq:DGF-odd-v}) in Lemma \ref{lem:DGF-closed} permits writing
\begin{eqnarray*}
B_M(s) &=& \sum_{\mbox{\footnotesize{odd }}j}\frac{v(j)^M}{j^s}+\sum_{l=1}\frac{2\nabla W_M(l)-(v(l)+v_2(l))^M}{(2l)^s}\\
&=& \left(1-\frac{1}{2^s}\right)V_M(s) + \frac{1}{2^{s-1}} B_M(s) - \frac{1}{2^s} Z_M(s).
\end{eqnarray*}
Solving for $B_M(s)$ proves the lemma.
\end{proof}

For $M=1$, applying (\ref{eq:DGF-v2}) from Lemma \ref{lem:DGF-closed} to Lemma \ref{lem:dgf-nabla-WM-exact} gives
\begin{eqnarray}\label{eq:B1s-closed}
B_1(s) &=& \frac{2^s-1}{2^s-2}V_1(s)-\frac{1}{2^s-2}Z_1(s)\nonumber\\
&=& \frac{2^s-1}{2^s-2}V_1(s)-\frac{1}{2^s-2} \left(V_1(s) + \frac{1}{2^s-1}\zeta(s)\right)\nonumber\\
&=& V_1(s) - \frac{1}{(2^s-1)(2^s-2)}\zeta(s).
\end{eqnarray}

Substituting this into (\ref{eq:TWn-integral-raw-dgf}) yields
\begin{equation}\label{eq:TW-integral-exact}
TW_1(n) = \frac{1}{2\pi i} \int_{3-i\infty}^{3+i\infty} V_1(s)\frac{n^s ds}{s(s+1)} - \frac{1}{2\pi i} \int_{3-i\infty}^{3+i\infty} \frac{\zeta(s)}{(2^s-1)(2^s-2)}\frac{n^s ds}{s(s+1)}.
\end{equation}

The second integral can be evaluated exactly by the method used in Section \ref{subsect:WDS1-evaluation}.
Evaluating the first integral requires more work.

Historically, $v(n)$ was one of the first digital functions to be analyzed using the Mellin transform techniques.
The original analysis in 1975 by Delange \cite{DELAN-1975} used
a combinatorial decomposition of the binary representations of integers
to directly derive an exact Fourier series formula for $\sum_{j<n} v(j)$.
In 1994, Flajolet et. al. \cite{FGKPT-1994} reproved Delange's result
using the Mellin transform techniques.
However, $V_1(s)$, the DGF of $v(n)$, does not seem to have been explicitly studied before Hwang's analysis \cite{HWANG-1998} in 1998.
First, denote $I_r(s)$ by
\begin{equation}\label{eq:def-Irs}
I_r(s) := -\frac{1}{2}\sum_{i=1}^\infty v(i)^r\left(\frac{1}{(2i)^s}-\frac{2}{(2i+1)^s}+\frac{1}{(2i+2)^s}\right).
\end{equation}
Standard algebraic manipulations, e.g. in \cite[pp.536]{HWANG-1998}, let us rewrite a summation of this form as an integral:
\begin{equation}\label{eq:Irs-integral}
I_r(s) = \frac{s}{2^s}\int_1^\infty\frac{v(\lfloor x\rfloor)^r}{x^{s+1}}\xi(x) dx,
\end{equation}
where
$$\xi(x) =
\begin{cases}
-\frac{1}{2} & , \mbox{if $\lfloor x \rfloor \leq x < \lfloor x \rfloor + \frac{1}{2},$} \\
\frac{1}{2} & , \mbox{if $\lfloor x \rfloor + \frac{1}{2} \leq x < \lfloor x \rfloor + 1.$}
\end{cases}$$

Hwang \cite{HWANG-1998} derived the following formula of $V_1(s)$,  revealing its singularities in $\Re(s)>-1$:
\begin{equation}\label{eq:dgf-V(s)-hwang}
V_1(s) = \frac{2^s-1}{2^s-2}\zeta(s) - \frac{1}{2(2^s-1)}\zeta(s) + \frac{2^s}{2^s-2} I_1(s),
\end{equation}

Substituting (\ref{eq:dgf-V(s)-hwang}) into (\ref{eq:B1s-closed}) yields
\begin{equation}\label{eq:B1s}
B_1(s) = \frac{2^{s+1}-1}{2(2^s-1)}\zeta(s)+\frac{2^s}{2^s-2}I_1(s).
\end{equation}
From the integral form of $I_1(s)$, we know that it is analytic in $\Re(s)>-1$. Hence, (\ref{eq:B1s}) (together with the fact that $\zeta(s)$ has no zero on the line $\Re(s)=0$) shows that at $s=0,\beta_j$, $B_1(s)$ possesses simple poles. Depending upon the values of $I_1(1)$ and $I_1(\alpha_j)$, $B_1(s)$ may either possess simple poles at $s=1,\,\alpha_j$ or be analytic at $s=1,\,\alpha_j$. These are all possible poles of $B_1(s)$ inside $\Gamma$ which we defined in (\ref{eq:Gamma}).

Using Hwang's representation {\em would} yield a closed-form formula for $TW_1(n)$ by considering contour $\Gamma$.
Unfortunately, the residues appearing in the resulting  Fourier coefficients
would be expressed in terms of the value of  $I_1(s)$ at various poles,
something which is not well understood.
In the next subsection, we will show how to use the higher order version of the Mellin-Perron formula to sidestep this issue and express the Fourier coefficients in terms of the Riemann-Zeta function.

\subsection{Moving Up to a Higher Order Case of the Mellin-Perron Formula}

We now see how to manipulate the first integral in (\ref{eq:TW-integral-exact}) to yield a formula in terms of values of the Riemann-Zeta function.
 
The general approach  is to note that  $V_1(s)$ is the DGF of $v(j)$, so the first integral in (\ref{eq:TW-integral-exact}), when transformed from integral back to summation by (\ref{eq:MPF-m=1}), is a double summation of $v(j)$. A double summation of $v(j)$ is also a triple summation of $\nabla v(j)$, and we can write a closed-form formula for the DGF of $\nabla v(j)$ in terms of $\zeta(s)$.
Equation (\ref{eq:MPF-m=2}) then provides  an  exact formula of the triple summation of $\nabla v(j)$, and we can evaluate the first integral in (\ref{eq:TW-integral-exact}).

We now present the details. Define
$$TV(n) := \frac{1}{n}\sum_{j=1}^{n}\sum_{i=1}^{j-1}v(i).$$
Algebraic manipulations permit 
writing $TV(n)$ in 
two  different ways:
\begin{equation}\label{eq:TTVn-double-summation}
TV(n) = \frac{1}{n}\sum_{k<n} v(k) (n-k),
\end{equation}
and
\begin{equation}\label{eq:TTVn-triple-summation}
TV(n) = \frac{1}{n}\sum_{k<n} \nabla v(k) \left[ \frac{(n-k)^2+(n-k)}{2}\right].
\end{equation}

Applying (\ref{eq:MPF-m=1}) to (\ref{eq:TTVn-double-summation}), yields
\begin{equation}\label{eq:TTVn-double-summation-integral}
TV(n) = \frac{1}{2\pi i}\int_{3-i\infty}^{3+i\infty} V_1(s)\frac{n^s ds}{s(s+1)},
\end{equation}
where the right side is exactly the first integral in (\ref{eq:TW-integral-exact}).

Applying (\ref{eq:MPF-m=2}) and (\ref{eq:MPF-m=1}) to (\ref{eq:TTVn-triple-summation}) gives the alternate expression
\begin{eqnarray}\label{eq:TTVn-triple-summation-integral}
&&TV(n)\\
&=& \frac{n}{2\pi i}\int_{3-i\infty}^{3+i\infty} \left(\sum_{j=1}^{\infty} \frac{\nabla v(j)}{j^s}\right) \frac{n^s ds}{s(s+1)(s+2)} + \frac{1}{4\pi i} \int_{3-i\infty}^{3+i\infty} \left(\sum_{j=1}^{\infty} \frac{\nabla v(j)}{j^s}\right)\frac{n^s ds}{s(s+1)}.\nonumber
\end{eqnarray}

Setting (\ref{eq:TTVn-double-summation-integral}) equal to (\ref{eq:TTVn-triple-summation-integral}) and using the closed-form formula for $\sum_{j=1}^{\infty}\nabla v(j)j^{-s}$ in Lemma \ref{lem:DGF-closed} gives
\begin{eqnarray*}
&&\frac{1}{2\pi i}\int_{3-i\infty}^{3+i\infty} V_1(s)\frac{n^s ds}{s(s+1)}\\
&=& \frac{n}{2\pi i} \int_{3-i\infty}^{3+i\infty} \frac{2^s-2}{2^s-1} \zeta(s)\frac{n^s ds}{s(s+1)(s+2)} + \frac{1}{4\pi i} \int_{3-i\infty}^{3+i\infty} \frac{2^s-2}{2^s-1} \zeta(s)\frac{n^s ds}{s(s+1)}.
\end{eqnarray*}

Substituting the above equality into (\ref{eq:TW-integral-exact}) yields a ``nicer''
integral representation for $TW_1(n)$.
\begin{eqnarray}\label{eq:TW-integral-exact-zeta-only}
TW_1(n) & = & \frac{n}{2\pi i} \int_{3-i\infty}^{3+i\infty} \frac{2^s-2}{2^s-1} \zeta(s)\frac{n^s ds}{s(s+1)(s+2)} \nonumber\\
& & + \frac{1}{4\pi i} \int_{3-i\infty}^{3+i\infty} \frac{2^s-2}{2^s-1} \zeta(s)\frac{n^s ds}{s(s+1)} \nonumber\\
& & - \frac{1}{2\pi i} \int_{3-i\infty}^{3+i\infty} \frac{1}{(2^s-1)(2^s-2)}\zeta(s)\frac{n^s ds}{s(s+1)}.
\end{eqnarray}

\subsection{Evaluation of Integrals}\label{subsect:WDS2_higher-evaluation}

The three integrals in (\ref{eq:TW-integral-exact-zeta-only}) can be evaluated almost exactly as in Section \ref{subsect:WDS1-evaluation}. That is, for the first integral consider contour $\Gamma ' = \Gamma '_1\cup\Gamma '_2\cup\Gamma '_3\cup\Gamma '_4$, where
\begin{eqnarray}\label{eq:Gamma2}
&\Gamma '_1 = \{ 3+iy : -R \leq y \leq R \},
\quad\quad\quad&
\Gamma '_2 = \{ x+iR : -5/4 \leq x \leq 3 \}, \nonumber\\
&\Gamma '_3 = \{ -5/4+iy : -R \leq y \leq R \},\quad
&\Gamma '_4 = \{ x-iR : -5/4 \leq x \leq 3 \}.
\end{eqnarray}
For the second and the third integrals consider contour $\Gamma$ defined in (\ref{eq:Gamma}). Next, prove that the integrals along the left, top and bottom paths tend to zero (using Lemma \ref{thm:bound_top_bottom_zeta_general} and Lemma \ref{thm:bound_left_zeta_general}). Finally, evaluate the sum of residues at the poles inside $\Gamma '$ or $\Gamma$. Since these are almost exactly the same as in Section \ref{subsect:WDS1-evaluation}, we leave out the details, only stating the results. See Figure \ref{fig:contour_WDS2_higher} for the contours.

\begin{figure}[t]
\vspace*{-.1in}
\centering%
\scalebox{0.31}{\includegraphics{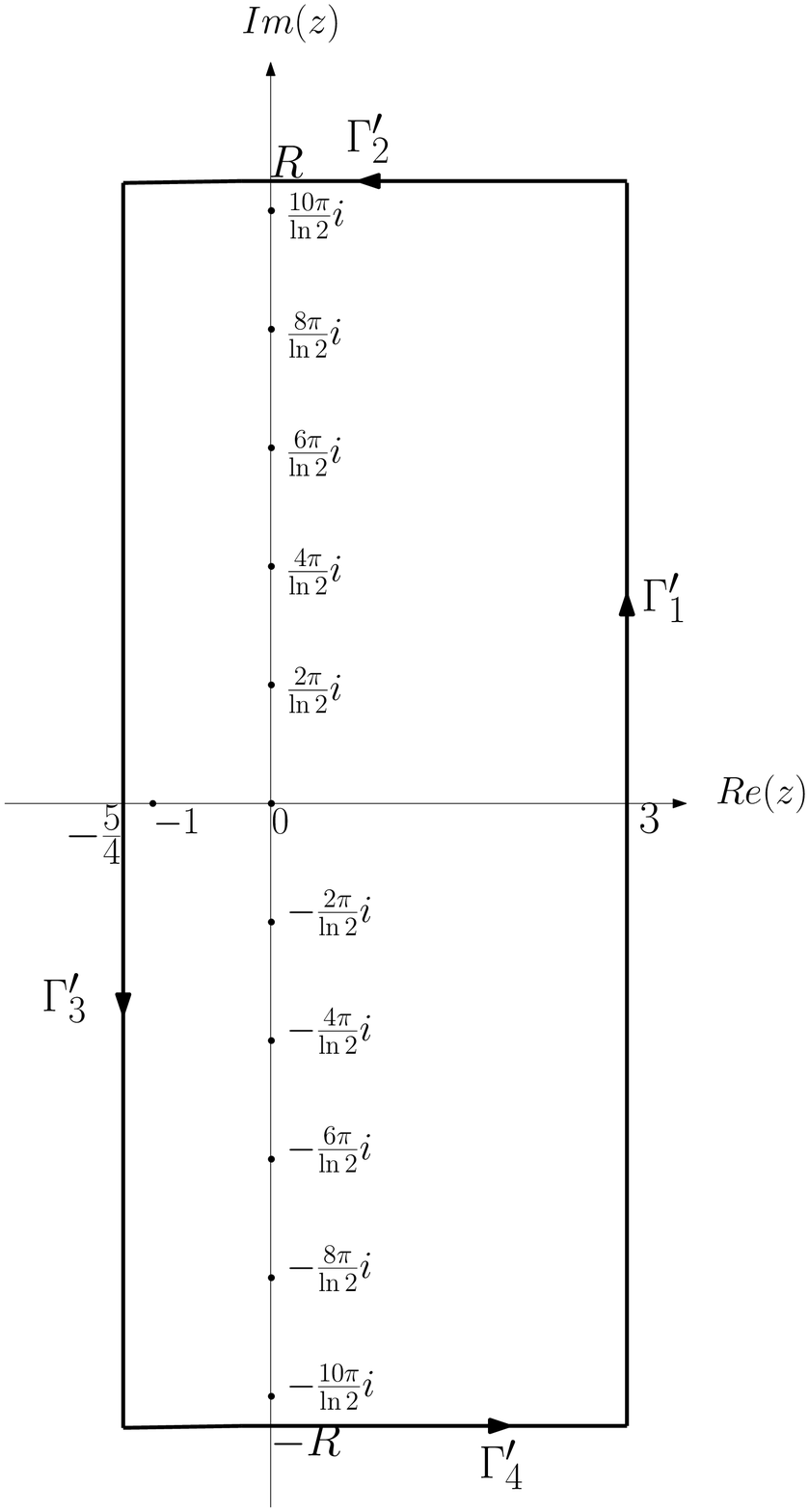}}
\scalebox{0.31}{\includegraphics{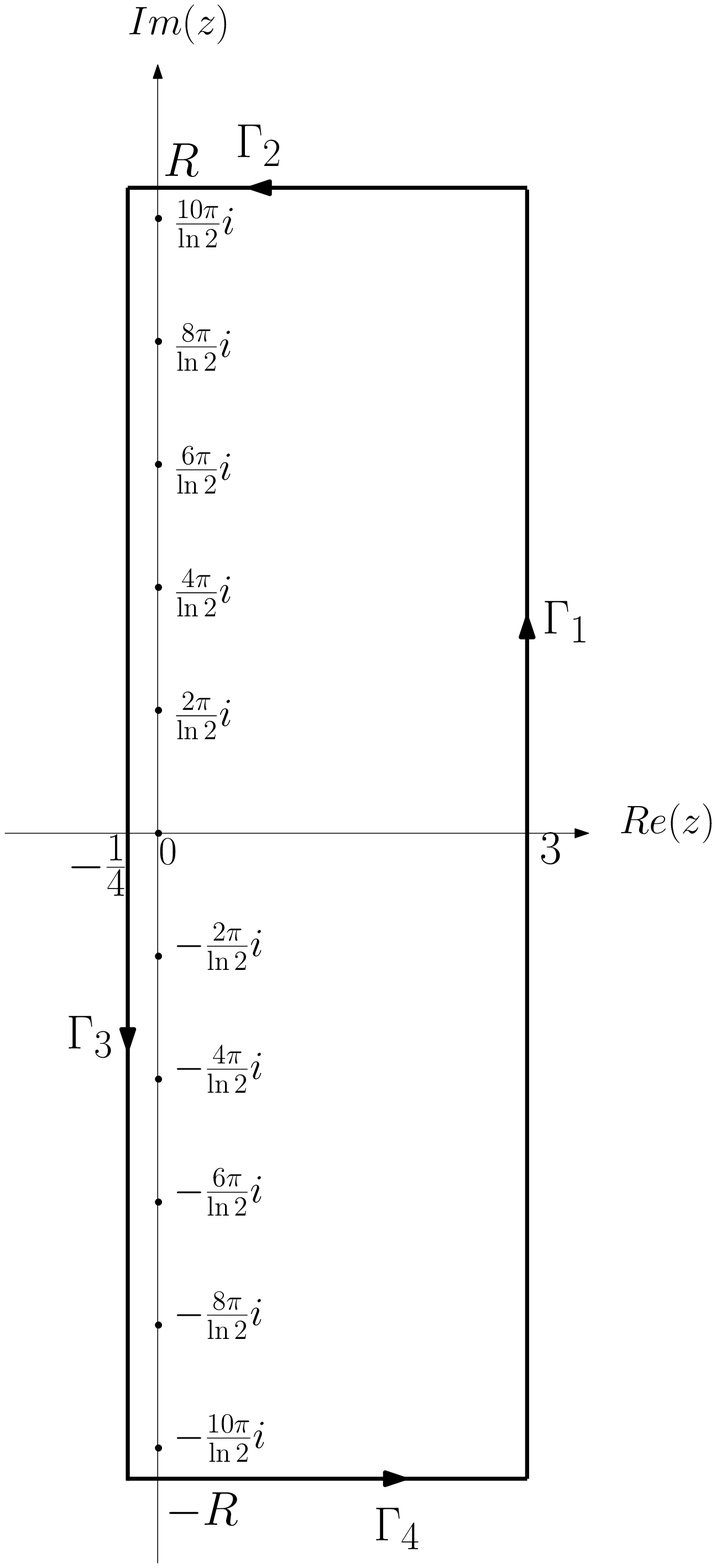}}
\scalebox{0.31}{\includegraphics{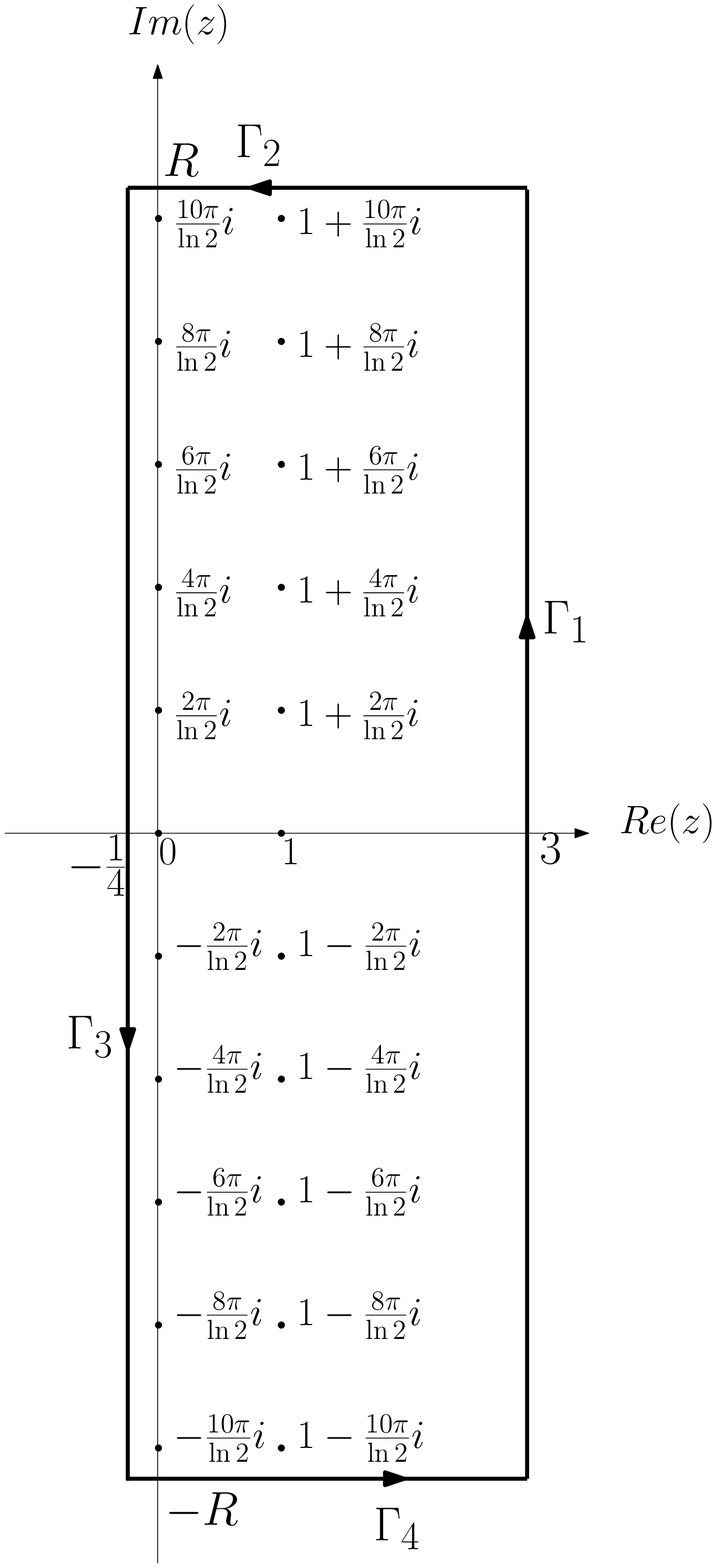}}
\caption{The coutours we use when evaluating the three integrals in (\ref{eq:TW-integral-exact-zeta-only}), with poles marked with dots. The left one is $\Gamma '$ defined in (\ref{eq:Gamma2}), which is for evaluating the first integral in (\ref{eq:TW-integral-exact-zeta-only}). The middle and right ones are both $\Gamma$ defined in (\ref{eq:Gamma}), which are for evaluating the second and third integrals in (\ref{eq:TW-integral-exact-zeta-only}) respectively.}
\label{fig:contour_WDS2_higher}
\end{figure}

The poles of the kernel of the the first integral inside $\Gamma '$ are a double pole at $s=0$ and simple poles at $s=\beta_j$ (where $j\in\mathbb{Z}\setminus\{0\}$) and $s=-1$. 
By summing the residues at all these poles, the first integral evaluates to 
\begin{equation}\label{eq:TW-first-integral-exact}
\frac{n}{2\pi i} \int_{3-i\infty}^{3+i\infty} \frac{2^s-2}{2^s-1}\zeta(s)\frac{n^s ds}{s(s+1)(s+2)} = \frac{1}{4} n \lg n + n H_1(\lg n) + \frac{1}{4},
\end{equation}
where
\begin{equation}\label{eq:w1f}
H_1(u) = \left(\frac{2\ln \pi - \ln 2 - 3}{8\ln 2}\right) - \frac{1}{\ln 2}\sum_{j\in\mathbb{Z}\setminus\{0\}} \frac{\zeta(\beta_j)}{\beta_j(\beta_j+1)(\beta_j+2)} e^{2\pi ij u}.
\end{equation}

\medskip

The poles of the kernel of the second integral inside $\Gamma$ are a double pole at $s=0$ and simple poles at $s=\beta_j$ (where $j\in\mathbb{Z}\setminus\{0\}$). 
By summing the residues at all these poles, the second integral evaluates to
\begin{equation}\label{eq:TW-second-integral-exact}
\frac{1}{4\pi i} \int_{3-i\infty}^{3+i\infty} \frac{2^s-2}{2^s-1} \zeta(s)\frac{n^s ds}{s(s+1)} = \frac{1}{4} \lg n + H_2(\lg n),
\end{equation}
where
\begin{equation}\label{eq:w2f}
H_2(u) = \left(\frac{2\ln\pi -\ln 2-2}{8\ln 2}\right) - \frac{1}{2\ln 2}\sum_{j\in\mathbb{Z}\setminus\{0\}} \frac{\zeta(\beta_j)}{\beta_j(\beta_j+1)} e^{2\pi ij u}.
\end{equation}

\medskip

The poles of the kernel of the third integral inside $\Gamma$ are 
a double pole at $s=1$, 
simple poles at $s=\alpha_j$ (where $j\in\mathbb{Z}\setminus\{0\}$), 
a double pole at $s=0$ and 
simple poles at $s=\beta_j$ (where $j\in\mathbb{Z}\setminus\{0\}$). 
By summing the residues at all these poles, the third integral evaluates to
\begin{equation}\label{eq:TW-third-integral-exact}
\frac{1}{2\pi i} \int_{3-i\infty}^{3+i\infty} \frac{1}{(2^s-1)(2^s-2)}\zeta(s)\frac{n^s ds}{s(s+1)} = \frac{1}{4}n\lg n + n H_{3,2}(\lg n) + \frac{1}{2}\lg n + H_{3,1}(\lg n),
\end{equation}
where
\begin{equation}\label{eq:w31f}
H_{3,1}(u) = \left(\frac{2\ln\pi+3\ln 2-2}{4\ln 2}\right) - \frac{1}{\ln 2}\sum_{j\in\mathbb{Z}\setminus\{0\}} \frac{\zeta(\beta_j)}{\beta_j(\beta_j+1)} e^{2\pi ij u}
\end{equation}
and
\begin{equation}\label{eq:w32f}
H_{3,2}(u) = \left(\frac{2\gamma_0-3-5\ln 2}{8\ln 2}\right) + \frac{1}{2\ln 2}\sum_{j\in\mathbb{Z}\setminus\{0\}} \frac{\zeta(\alpha_j)}{\alpha_j(\alpha_j+1)} e^{2\pi ij u}.
\end{equation}

\medskip

Combining the three integrals above yields:
\begin{Theorem}\label{thm:TWn-general}
\begin{equation}\label{eq:TWn-general}
TW_1(n) = n F_{W,1}(\lg n) - \frac{1}{4} \lg n + F_{W,0}(\lg n).
\end{equation}
where 
$F_{W,1}(u)$ and $F_{W,0}(u)$ are two absolutely convergent Fourier series, whose
coefficients are given by 
\begin{eqnarray*}
F_{W,1}(u) 
&=& \frac{\ln\pi-\gamma_0+2\ln 2}{4\ln 2} -\frac{1}{2\ln 2}\sum_{j\in\mathbb{Z}\setminus\{0\}}\left(\frac{2\zeta(\beta_j)}{\beta_j(\beta_j+1)(\beta_j+2)} + \frac{\zeta(\alpha_j)}{\alpha_j(\alpha_j+1)}\right)e^{2\pi iju}\\
F_{W,0}(u) &=& 
\frac{2-2\ln\pi-5\ln 2}{8\ln 2} + \frac{1}{2\ln 2}\sum_{j\in\mathbb{Z}\setminus\{0\}}\frac{\zeta(\beta_j)}{\beta_j(\beta_j+1)}e^{2\pi iju}.
\end{eqnarray*}
The average value of $F_{W,1}(u)$ is
$$\frac{\ln\pi-\gamma_0+2\ln 2}{4\ln 2}\approx 0.704687.$$
\end{Theorem}
\begin{proof}
Substituting (\ref{eq:TW-first-integral-exact}), (\ref{eq:TW-second-integral-exact}) and (\ref{eq:TW-third-integral-exact}) into (\ref{eq:TW-integral-exact-zeta-only}) yields
$$TW_1(n) = n\left[H_1(\lg n)-H_{3,2}(\lg n)\right]-\frac{1}{4}\lg n + \left[H_2(\lg n)-H_{3,1}(\lg n)+\frac{1}{4}\right].$$
Setting $F_{W,1}(u) := H_1(u)-H_{3,2}(u)$ and $F_{W,0}(u) := H_2(u)-H_{3,1}(u)+\frac{1}{4}$ gives  (\ref{eq:TWn-general}) and the Fourier series.

Lemma \ref{lem:zeta-bound} gives
$$\zeta(\beta_j) = O(|j|^{1/2}\log |j|)\quad\mbox{and}\quad\zeta(\alpha_j) = O(\log |j|).$$
Hence, as 
$|j|\rightarrow\infty$
the terms in (\ref{eq:w1f}), (\ref{eq:w2f}) and (\ref{eq:w31f}) are $O(|j|^{-3/2}\log |j|)$ and the terms in (\ref{eq:w32f}) are $O(|j|^{-2}\log |j|)$, implying the absolute convergences of $H_1(u)$, $H_2(u)$, $H_{3,1}(u)$ and $H_{3,2}(u)$, and thus $F_{W,1}(u)$ and $F_{W,0}(u)$.
\end{proof}

\section{More Weighted Digital Sums of the Second Type}\label{sect:WDS2}

We now analyze $TW_M(n) = \sum_{j<n} W_M(j)$ as defined by (\ref{eq:def-WMn}) and (\ref{eq:def-TWMn}). 
Again, by  Lemma \ref{thm:key}, 
\begin{equation}\label{eq:TWMn-integral-raw-dgf}
TW_M(n) = \frac{1}{2\pi i}\int_{c-i\infty}^{c+i\infty}B_M(s) \frac{n^s ds}{s(s+1)},
\end{equation}
where $B_M(s)$ is the DGF of $\nabla W_M(j)$ as defined in (\ref{eq:def-BMs}).

As before, we will integrate along contour $\Gamma$ we defined in (\ref{eq:Gamma}), and prove that the integrals along the top, bottom and left contours vanish as $R\rightarrow\infty$, while that on the right contour equals (\ref{eq:TWMn-integral-raw-dgf}) and 
then apply the Cauchy residue theorem.

The DGF $B_M(s)$ for $M >1$ is much more complicated than the DGFs previously encountered in this paper. We will therefore have to introduce new techniques to study it.

\subsection{Properties of Poles of the DGF}\label{subsect:WDS2-DGF}

We saw from (\ref{eq:B1s}) that the order of the poles of $\frac{B_1(s)n^s}{s(s+1)}$
at $s=1$ and $s=\alpha_j$ were all at most $1$.
By Lemma \ref{thm:sum-residues}, this implied that
``coefficient'' of the first order term of $TW_1(n)$, i.e. the $n$ term in (\ref{eq:TWn-general}),
is the Fourier series $F_{W,1}(\lg n)$.
Analogously, for all $M>1$ we will prove that the orders of poles of
$\frac{B_M(s)n^s}{s(s+1)}$ at $s=1$ and $s=\alpha_j$ are all at most $1$.
This will again imply that for $M >1$, $TW_M(n)$ has a
``periodic first-order coefficient''.

To start, we will need the following semi-recursive formula of $B_M(s)$.
\begin{Lemma}\label{lem:BMs-funeq}
For $M>1$,
$B_M(s)$ satisfies 
\begin{equation}\label{eq:BMs-key3}
B_M(s) = \frac{2^s}{2^s-1}\zeta(s)+\frac{1}{2^s-1}\sum_{r=1}^{M-1}\binom{M}{r}B_r(s)
-\frac{2^s}{2^s-2}\sum_{r=1}^M\binom{M}{r}R_r(s),
\end{equation}
where $R_r(s)$ is defined as
\begin{equation}\label{eq:def-Rrs}
R_r(s) := \sum_{i=1}^\infty v(i)^r\left[\frac{1}{(2i)^s}-\frac{1}{(2i+1)^s}\right].
\end{equation}
\end{Lemma}

\begin{proof}
Lemma \ref{lem:dgf-nabla-WM-exact} gives
\begin{equation}\label{eq:BMS-closed-recall}
B_M(s)=\frac{2^s-1}{2^s-2}V_M(s)-\frac{1}{2^s-2}Z_M(s).
\end{equation}
We now derive two seperate functional equations for  $V_M(s)$ and $Z_M(s)$ and combine them to yield (\ref{eq:BMs-key3}).

\medskip

Recalling from Lemma \ref{lem:property-v-and-v_2} that $v(2i)=v(i)$ and $v(2i+1)=v(i)+1$ gives
\begin{eqnarray*}
V_M(s) &=& \sum_{i\geq 1}\frac{v(i)^M}{(2i)^s} + \sum_{i\geq 0}\frac{(v(i)+1)^M}{(2i+1)^s}\\
&=& \frac{1}{2^s}V_M(s)+1+\sum_{r=1}^M \binom{M}{r}\left(\sum_{i=1}^\infty\frac{v(i)^r}{(2i+1)^s}\right) + \sum_{i=1}^\infty\frac{1}{(2i+1)^s}\\
&=& \frac{1}{2^s}V_M(s)+\sum_{r=1}^M \binom{M}{r}\left(\frac{1}{2^s}V_r(s)-R_r(s)\right) + \left(1-\frac{1}{2^s}\right)\zeta(s).
\end{eqnarray*}

Solving for $V_M(s)$ yields
\begin{equation}\label{eq:VMs-funeq}
V_M(s) = \frac{2^s-1}{2^s-2}\zeta(s) + \frac{1}{2^s-2}\sum_{r=1}^{M-1} \binom{M}{r}V_r(s)-\frac{2^s}{2^s-2}\sum_{r=1}^M\binom{M}{r}R_r(s).
\end{equation}

\medskip

Next, from Lemma \ref{lem:property-v-and-v_2} it is easy to show $v(2i)+v_2(2i)=v(i)+v_2(i)+1$ and $v(2i+1)+v_2(2i+1)=v(2i+1)$. Also, using (\ref{eq:DGF-odd-v}) in Lemma \ref{lem:DGF-closed}, gives
\begin{eqnarray*}
Z_M(s) &=& \sum_{i\geq 1}\frac{(v(i)+v_2(i)+1)^M}{(2i)^s} + \sum_{\mbox{\footnotesize{odd }}j}\frac{v(j)^M}{j^s}\\
&=& \frac{1}{2^s}\sum_{r=0}^M\binom{M}{r}Z_r(s)+\left(1-\frac{1}{2^s}\right)V_M(s).
\end{eqnarray*}

Since $Z_0(s)=\zeta(s)$, this solves to
\begin{equation}\label{eq:ZMs-funeq}
Z_M(s) = V_M(s) + \frac{1}{2^s-1}\zeta(s)+\frac{1}{2^s-1}\sum_{r=1}^{M-1}\binom{M}{r}Z_r(s).
\end{equation}

\medskip

Substituting (\ref{eq:VMs-funeq}) and (\ref{eq:ZMs-funeq}) into (\ref{eq:BMS-closed-recall}) yields
\begin{equation}\label{eq:BMs-key2}
B_M(s) = \frac{2^s}{2^s-1}\zeta(s)+\frac{1}{2^s-1}\sum_{r=1}^{M-1}
\binom{M}{r}\frac{(2^s-1)V_r(s)-Z_r(s)}{2^s-2} -\frac{2^s}{2^s-2}\sum_{r=1}^M\binom{M}{r}R_r(s).
\end{equation}
Finally, using (\ref{eq:BMS-closed-recall}) to simplify the internal terms in  (\ref{eq:BMs-key2}) yields (\ref{eq:BMs-key3}).
\end{proof}

The next lemma gives a ``closed-form'' formula for $B_M(s)$, in terms of $I_k(s)$, previously  defined in (\ref{eq:def-Irs}), and $\zeta(s)$.

\begin{Lemma}
\begin{equation}\label{eq:BMS-form}
B_M(s) = \frac{P_{M,1}(2^s)}{(2^s-1)^M}\zeta(s)
+ \sum_{k=1}^M \frac{P_{M,2,k}(2^s)}{(2^s-1)^{M-k}(2^s-2)} I_k(s),
\end{equation}
where $P_{M,1}(x)$ and $P_{M,2,k}(x)$ are two polynomials, with $P_{1,1}(1)=\frac{1}{2}$, $P_{M,1}(1) = M P_{M-1,1}(1)-M!/2^M$ for $M\geq 2$ and $P_{M,2,k}(0) = 0$ for $k=1,2,\cdots,M$.
\end{Lemma}

\begin{proof}
We prove by induction, using Lemma \ref{lem:BMs-funeq}. 
First, note that (\ref{eq:B1s}) is just the special case of this Lemma for $M=1$. It also gives $P_{1,1}(1)=\frac{1}{2}$.

Now assume the lemma is true for $M<M_0$.  Then,
\begin{eqnarray}
&& \frac{1}{2^s-1}\sum_{r=1}^{M_0-1} \binom{M_0}{r} B_r(s)\nonumber\\
&=& \left(\sum_{r=1}^{M_0-1} \frac{\binom{M_0}{r} P_{r,1}(2^s)}{(2^s-1)^{r+1}}\right)\zeta(s) + \frac{1}{2^s-1} \sum_{r=1}^{M_0-1}\sum_{k=1}^r \frac{\binom{M_0}{r}P_{r,2,k}(2^s)}{(2^s-1)^{r-k}(2^s-2)} I_k(s)\nonumber\\
&=& \left(\sum_{r=1}^{M_0-1} \frac{\binom{M_0}{r} P_{r,1}(2^s)}{(2^s-1)^{r+1}}\right)\zeta(s) + \sum_{k=1}^{M_0-1} I_k(s) \sum_{r=k}^{M_0-1} \frac{\binom{M_0}{r}P_{r,2,k}(2^s)}{(2^s-1)^{r+1-k}(2^s-2)}\nonumber\\
&:=& \frac{T_{M_0,1}(2^s)}{(2^s-1)^{M_0}}\zeta(s) + \sum_{k=1}^{M_0-1} \frac{T_{M_0,2,k}(2^s)}{(2^s-1)^{M_0-k}(2^s-2)} I_k(s),\label{eq:BMs-form-step1}
\end{eqnarray}
where $T_{M_0,1}(x)$ and $T_{M_0,2,k}(x)$'s are polynomials satisfying $T_{M_0,1}(1) = \binom{M_0}{M_0-1} P_{M_0-1,1}(1) = M_0 P_{M_0-1,1}(1)$ and $T_{M_0,2,k}(0) = 0$ for $k=1,2,\cdots,M_0-1$.

Now set
$$D_r(s) := \frac{1}{2^{s+1}}\sum_{i=1}^\infty \frac{\nabla [v(i)^r]}{i^s}.$$
From the definition of $I_r(s)$ in
 (\ref{eq:def-Irs}) and  some algebraic manipulation, 
$$R_r(s) = D_r(s) - I_r(s).$$

Grabner and Hwang \cite{GrHw-2005} proved
$$D_r(s) = \frac{2^s-2}{2^s}\zeta(s)\sum_{k=1}^r\frac{k!S(r,k)}{2^k(2^s-1)^k}
+\sum_{h=1}^{r-1}\binom{r}{h}I_{r-h}(s)\sum_{k=1}^h\frac{k!S(h,k)}{2^k(2^s-1)^k},
$$
where $S(n,k)$ are the Stirling numbers of the second kind.
Noting that $S(n,n)=1$, $R_r(s)$ can be rewritten as
$$R_r(s) = \frac{2^s-2}{2^s}\frac{Q_{r,1}(2^s)}{(2^s-1)^r}\zeta(s) + \sum_{k=1}^r  \frac{Q_{r,2,k}(2^s)}{(2^s-1)^{r-k}} I_k(s),$$
where $Q_{r,1}(x)$ and $Q_{r,2,k}(x)$'s are polynomials satisfying $Q_{r,1}(1) = \frac{r! S(r,r)}{2^r} = r!/2^r$.

This permits writing
\begin{eqnarray}
&& \frac{2^s}{2^s-2}\sum_{r=1}^{M_0} \binom{M_0}{r} R_r(s)\nonumber\\
&=& \left(\sum_{r=1}^{M_0} \frac{\binom{M_0}{r} Q_{r,1}(2^s)}{(2^s-1)^r}\right)\zeta(s) + \frac{2^s}{2^s-2}\sum_{r=1}^{M_0} \sum_{k=1}^r \frac{\binom{M_0}{r}Q_{r,2,k}(2^s)}{(2^s-1)^{r-k}} I_k(s)\nonumber\\
&=& \left(\sum_{r=1}^{M_0} \frac{\binom{M_0}{r} Q_{r,1}(2^s)}{(2^s-1)^r}\right)\zeta(s) + \frac{2^s}{2^s-2} \sum_{k=1}^{M_0} I_k(s) \sum_{r=k}^{M_0}\frac{\binom{M_0}{r}Q_{r,2,k}(2^s)}{(2^s-1)^{r-k}}\nonumber\\
&:=& \frac{U_{M_0,1}(2^s)}{(2^s-1)^{M_0}}\zeta(s) + \frac{2^s}{2^s-2}\sum_{k=1}^{M_0} \frac{U_{M_0,2,k}(2^s)}{(2^s-1)^{M_0-k}} I_k(s),\label{eq:BMs-form-step2}
\end{eqnarray}
where $U_{M_0,1}(x)$ and $U_{M_0,2,k}(x)$'s are polynomials satisfying $U_{M_0,1}(1) =  M_0! / 2^{M_0}$.

Substituting (\ref{eq:BMs-form-step1}) and (\ref{eq:BMs-form-step2}) back into (\ref{eq:BMs-key3}) gives
\begin{eqnarray*}
B_M(s) &=& \frac{2^s}{2^s-1}\zeta(s) + \left(\frac{T_{M_0,1}(2^s)}{(2^s-1)^{M_0}}\zeta(s) + \sum_{k=1}^{M_0-1} \frac{T_{M_0,2,k}(2^s)}{(2^s-1)^{M_0-k}(2^s-2)} I_k(s)\right)\\
&&\qquad\qquad- \left(\frac{U_{M_0,1}(2^s)}{(2^s-1)^{M_0}}\zeta(s) + \sum_{k=1}^{M_0} \frac{2^s U_{M_0,2,k}(2^s)}{(2^s-1)^{M_0-k}(2^s-2)} I_k(s)\right).
\end{eqnarray*}
The lemma is proved for $M=M_0$ by setting $P_{M_0,1}(x) := x(x-1)^{M_0-1} + T_{M_0,1}(x) - U_{M_0,1}(x)$, $P_{M_0,2,k}(x) := T_{M_0,2,k}(x) - x U_{M_0,2,k}(x)$ for $k=1,2,\cdots,M_0-1$ and $P_{M_0,2,M_0}(x) = -x U_{M_0,2,M_0}(x)$.
\end{proof}

We can now find the poles of $B_M(s)$ inside $\Gamma$. See Figure \ref{fig:contour_gamma2} for locations.
\begin{Corollary}\label{thm:BMs-asym-behaviour}
For $M\geq 1$, The singularities of $B_M(s)$ inside $\Gamma$ are\\
(i) poles of order at most $1$ at $s=1$ and $s=\alpha_j$; and\\
(ii) poles of order $M$ at $s=0$ and $s=\beta_j$. 

Hence, the singularities of $\frac{B_M(s)n^s}{s(s+1)}$ inside $\Gamma$ are\\
(i) poles of order at most $1$ at $s=1$ and $s=\alpha_j$;\\
(ii)  a pole of order $M+1$ at $s=0$;  and\\
(iii)  poles of order $M$ at $s=\beta_j$.
\end{Corollary}
\begin{proof}
(\ref{eq:BMS-form}) permits us to identify the singularities by working through the
various terms and recalling that $I_k(s)$ is analytic when $\Re(s) > -1.$

The recurrence relations  $P_{M,1}(1) = M P_{M-1,1}(1)-M!/2^M$ with initial condition $P_{1,1}(1)=1/2$ give $P_{M,1}(1)>0$ for $M\geq 1$. 
Hence at $s=0,\beta_j$, $P_{M,1}(2^s)/(2^s-1)^M$ has poles of order \emph{exactly} $M$, while $\zeta(s)$ is analytic (but is not  zero). 

At $s=\alpha_j,$ 
$P_{M,1}(2^s)/(2^s-1)^M$ and $\zeta(s)$ are all analytic.

At $s=1$,  $P_{M,1}(2^s)/(2^s-1)^M$ is analytic, but $\zeta(s)$ has a simple pole.

At $s=0,\beta_j$,
the order of poles of $P_{M,2,k}(2^s)(2^s-1)^{-(M-k)}(2^s-2)^{-1} I_k(s)$  is {\em at most} $M-1$. 

At $s=1,\alpha_j$, 
$P_{M,2,k}(2^s)(2^s-1)^{-(M-k)}(2^s-2)^{-1} I_k(s)$ has poles of order at most $1$  
(due to the term  $(2^s-2)^{-1}$).
\end{proof}

\begin{figure}[t]
\vspace*{-.1in}
\centering%
\scalebox{0.38}{\includegraphics{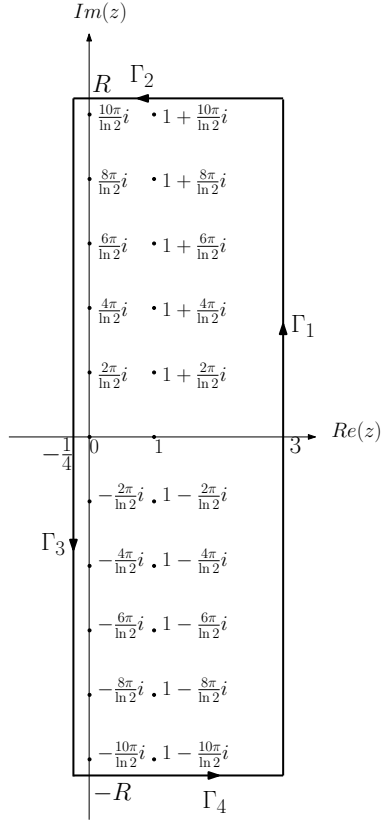}}
\caption[The Contour $\Gamma$ for evaluating $TW_M(n)$]{The figure is contour $\Gamma$ defined in (\ref{eq:Gamma}). The dots represent the poles of $\frac{B_M(s)n^s}{s(s+1)}$ inside $\Gamma$.}
\label{fig:contour_gamma2}
\end{figure}

\subsection{A Formula for $TW_M(n)$}\label{subsect:WDS2-evaluation}

As in the previous problems, we must again first show that the integrals along the top, bottom and left contours vanish as $R\rightarrow\infty$. 

We need two basic observations.
Suppose $H(s) = P(2^s)(2^s-1)^{-N_1}(2^s-2)^{-N_2}$, where $P$ is a polynomial and $N_1,N_2$ are non-negative integers.\\
\\
\textbf{\underline{Fact 1:}} When $\Re(s)<0$, $H(s)$ can be expressed as a power series of $2^s$, and this series is absolutely and uniformly convergent on the line $\Re(s)+(-\infty,+\infty)i$.\\[.05in]
Furthermore,  if $P(0)=0$, i.e. the constant term of $P$ is zero, then the constant term of the power series is also zero.\\[.1in]
\textbf{\underline{Fact 2:}} $H(s)$ is bounded along the line segment $(-1/4,3)+iR_j$ independently of $j$.

\begin{Lemma}\label{thm:bound_top_bottom_BMs}
$$\lim_{j\rightarrow\infty}\int_{-1/4+iR_j}^{3+iR_j} \frac{B_M(s)n^s}{s(s+1)}ds = 0.$$
\end{Lemma}
\begin{proof}
For $s\in\left(-1/4,3\right)+iR_j$, Grabner and Hwang \cite{GrHw-2005} proved that
$$|I_M(s)| = O(|j|^{3/4}\log^{2M}|j|) = o(|j|).$$
Furthermore,  Lemma \ref{lem:zeta-bound} gives
$$|\zeta(s)| = O(|j|^{3/4}\log |j|) = o(|j|).$$

By  (\ref{eq:BMS-form}) and Fact 2, $|B_M(s)|$ is bounded by $o(|j|)$ along  $\left(-1/4,3\right)+iR_j$. Hence
\begin{eqnarray*}
\left|\int_{-1/4+iR_j}^{3+iR_j}\frac{B_M(s)n^s}{s(s+1)}ds\right|
&\leq & \int_{-1/4+iR_j}^{3+iR_j} \left|\frac{B_M(s)n^s}{s(s+1)}\right| ds\\
&\leq & \int_{-1/4+iR_j}^{3+iR_j} \left(o(|j|)\times O(|j|^{-2})\times n^3\right)ds\\
&=& \frac{13}{4}o(|j|^{-1})n^3\\
&\rightarrow& 0
\end{eqnarray*}
as $j\rightarrow\infty$.
\end{proof}

\begin{Lemma}\label{thm:left=0-BMs}
For any positive integer $M$,
$$\int_{-1/4-i\infty}^{-1/4+i\infty}B_M(s) \frac{n^s ds}{s(s+1)} = 0.$$
\end{Lemma}
\begin{proof}
Grabner and Hwang \cite{GrHw-2005} proved the bound
$$\left|I_r\left(-\frac{1}{4}+it\right)\right| = O(|t|^{3/4}\log^{2r}|t|) = O(|t|^{3/4+\delta})$$
for any $\delta>0$. This upper bound allows us to use a theorem from Hwang \cite{HWANG-1998} to prove

Hwang \cite{HWANG-1998} proved the following theorem:
\begin{Theorem}\label{thm:left=0-general}
Suppose $U(s) = s2^{-s}\int_1^\infty u(x)\xi(x) x^{-s-1} dx$ for some nonnegative, real arithmetic function $u(x) = u_{\lfloor x\rfloor}$. If
\begin{enumerate}
\item $U(s)$ converges for $\Re(s)>\sigma_u$, where $\sigma_u<\sigma$,
\item $|U(\sigma+it)|=O(|t|^\delta)$ for some $0<\delta<1$,
\end{enumerate}
then we have
$$\frac{1}{2\pi i}\int_{\sigma-i\infty}^{\sigma+i\infty}\frac{(2^kn)^s}{s(s+1)} U(s)ds = 0$$
for all integers $k,n\geq 1$.
\end{Theorem}
Grabner and Hwang \cite{GrHw-2005} proved the bound
$$\left|I_r\left(-\frac{1}{4}+it\right)\right| = O(|t|^{3/4}\log^{2r}|t|) = O(|t|^{3/4+\delta})$$
for any $\delta>0$, which enables us to use Theorem \ref{thm:left=0-general} to get
\begin{equation}\label{eq:bound_left_Irs_general}
\int_{-1/4-i\infty}^{-1/4+i\infty}\frac{(2^k n)^s}{s(s+1)}I_r(s) ds = 0
\end{equation}
for positive integers $k,n,r$.

(\ref{eq:BMS-form}) shows that $B_M(s)$ can be expressed in the form of
$$B_M(s) = \frac{P_{M,1}(2^s)}{(2^s-1)^M}\zeta(s)
+ \sum_{k=1}^M \frac{P_{M,2,k}(2^s)}{(2^s-1)^{M-k}(2^s-2)} I_k(s),$$
while $P_{M,2,k}(0)=0$. By Fact 1, when $\Re(s)=-1/4$, $P_{M,1}(2^s)(2^s-1)^{-M}$ and $P_{M,2,k}(2^s)(2^s-1)^{-(M-k)}(2^s-2)^{-1}$ can be expressed as power series of $2^s$, and the power series for $P_{M,2,k}(2^s)(2^s-1)^{-(M-k)}(2^s-2)^{-1}$ have zero constant terms. Hence, when $\Re(s)=-1/4$, we may rewrite $B_M(s)$ to be
$$B_M(s) = \sum_{j=0}^\infty p_j 2^{js} \zeta(s) + \sum_{k=1}^M\sum_{j=1}^\infty q_{k,j} 2^{js} I_k(s)$$
for some $\{p_j\}$ and $\{q_{k,j}\}$. Hence
\begin{eqnarray*}
&& \int_{-1/4-i\infty}^{-1/4+i\infty} B_M(s) \frac{n^s ds}{s(s+1)}\\
&=& \int_{-1/4-i\infty}^{-1/4+i\infty} \left(\sum_{j=0}^\infty p_j (2^j n)^s \zeta(s) + \sum_{k=1}^M\sum_{j=1}^\infty q_{k,j} (2^j n)^s I_k(s)\right)\,\frac{ds}{s(s+1)}.
\end{eqnarray*}
However, the power series $\sum_{j=0}^\infty p_j (2^j n)^s$ and $\sum_{j=1}^\infty q_{k,j} (2^j n)^s$ are uniformly convergent on $-1/4+(-\infty,\infty)i$, by Fact 1. 
This allows interchange of the integral sign and the summation signs.

Hence, $\int_{-1/4-i\infty}^{-1/4+i\infty}B_M(s) \frac{n^s ds}{s(s+1)}$ can be expressed as a series, in which each term is either a constant multiplied by an  integral in the form of  (\ref{eq:bound_left_Irs_general}), or a constant multiplied by an  integral in the form of (\ref{eq:left=0-m=1}).
\end{proof}

We can now state our final result. 
\begin{Theorem}
\begin{equation}\label{eq:TWMn-general}
TW_M(n) = n G_M(\lg n) + d_M \lg^M n + \sum_{d=0}^{M-1} \left(\lg^d n\right)G_{M,d}(\lg n),
\end{equation}
where $d_M$ is a constant, $G_M(u)$ and $G_{M,d}(u)$'s are periodic functions with period one given by absolutely convergent Fourier series.
\end{Theorem}
\begin{proof}
Consider the contour $\Gamma$ in Figure \ref{fig:contour_gamma2}, taking $R\rightarrow\infty$. Lemma \ref{thm:bound_top_bottom_BMs} and Lemma \ref{thm:left=0-BMs} show that $\frac{1}{2\pi i}\int_{\Gamma_q}\frac{B_M(s)n^s}{s(s+1)}ds=0$ for $q=2,3,4$. Hence
$$TW_M(n) = \frac{1}{2\pi i} \int_{\Gamma_1} B_M(s) \frac{n^s}{s(s+1)}ds,$$
is the sum of residues at the poles of $\frac{B_M(s)n^s}{s(s+1)}$ inside $\Gamma$, by the Cauchy residue theorem.

By Lemma \ref{lem:zeta-bound}, we have the bound $|\zeta(\sigma+it)|=O(|t|^{1/2+\epsilon}\log|t|)$ when $\sigma\geq-\epsilon$ for sufficiently small $\epsilon$. Grabner and Hwang \cite{GrHw-2005} also proved that $|I_r(\sigma+it)|=O(|t|^{1/2+\epsilon}\log^{2r}|t|)$ when $\sigma\geq-\epsilon$ for sufficiently small $\epsilon$.
Hence by Lemma \ref{thm:bound-derivative}, 
$$|\zeta^{(q)}(\alpha_j)|,\ |\zeta^{(q)}(\beta_j)| = O(|j|^{1/2+\epsilon}\log|j|)$$
and
$$|I_r^{(q)}(\alpha_j)|,\ |I_r^{(q)}(\beta_j)| = O(|j|^{1/2+\epsilon}\log^{2r}|j|)$$
for any fixed positive integer $q$.

$B_M(s)$ can be expressed in the form of (\ref{eq:BMS-form}). Knowing that each function in the form of $P(2^s)(2^s-1)^{-N_1}(2^s-2)^{-N_2}$ will have a Laurent series with identical coefficients at $\theta_j=\sigma+\frac{2\pi j}{\ln 2}i$ for any fixed $\sigma$, togather with the results from the last paragraph and Corollary \ref{thm:BMs-asym-behaviour}, we use Lemma \ref{thm:sum-residues} when $\sigma=0,1$ to derive
$$\sum_{j\in\mathbb{Z}}\mbox{Res}\left(\frac{B_M(s) n^s}{s(s+1)},s=\alpha_j\right) = n G_M(\lg n)$$
and
$$\sum_{j\in\mathbb{Z}}\mbox{Res}\left(\frac{B_M(s) n^s}{s(s+1)},s=\beta_j\right) = d_M\lg^M n + \sum_{d=0}^{M-1} \left(\lg^d n\right)G_{M,d}(\lg n),$$
where $G_M(u)$ and $G_{M,d}(u)$'s are periodic functions with period one given by absolutely convergent Fourier series.
\end{proof}

\section{Conclusion}\label{sect:conclusion}

Mellin Transform techniques have previously  been extensively used to analyze various divide-and-conquer algorithms and digital sums.  A common theme in those analyses is the appearance of a (usually second order) periodic term, usually expressed in terms of a Fourier series.  This Fourier series is the sum of residues of a complex function which has singularities regularly spaced along a vertical line.

In this paper we pushed the technique further to derive exact analyses of the solution to
multidimensional divide-and-conquer recurrences and various, more complicated, weighted 
digital sums. Our closed form solutions had the properties that {\em all} terms were either polylogarithmic or $n$ times a polylogarithm, with all coefficients either being constant or 
a periodic function given by a Fourier series.

Our analysis of the multidimensional divide-and-conquer recurrence was a straightforward
extension of the use of Mellin transform techniques for  the analysis of simple
divide-and-conquer recurrences.   Our analyses of weighted digital sums,  though,  required 
developing a better understanding of various Dirichlet generating functions of differences
of digital functions.

\bibliographystyle{ieeetr}

\appendix

\section{Proof of Lemma \ref{lem:solution-f-integral}}\label{app:DC_solution}

Our main technique for solving the Multidimensional Divide-and-Conquer recurrence is a generalization of Lemma \ref{lem:solution-f-integral} for basic divide-and-conquer recurrences, originally proved in \cite{FlGo-1994} by Flajolet and Golin.
In order to make this paper self-contained, we provide the proof of that lemma here.

The divide-and-conquer recurrence is
$$f_n = f_{\lfloor n/2 \rfloor} + f_{\lceil n/2 \rceil} + e_n$$
with initial condition $f_1 = 0$ and given  ``conquer'' cost sequence $\{e_n\}$ 
where $e_0=e_1=0$.

Distinguishing between odd and even cases of the recurrence, we find that for $j\geq 1$,
\begin{equation}\label{eq:def_D_and_C}
f_{2j}  = 2 f_j + e_{2j},
\qquad
f_{2j+1}  = f_j + f_{j+1} + e_{2j+1}.
\end{equation}

Let $\nabla g_n = g_n - g_{n-1}$ be the backward difference operator. Then, for $j \geq 1$,
\begin{equation}
\label{eq:DC_back}
\nabla f_{2j}  =  \nabla f_j + \nabla e_{2j},
\quad
\nabla f_{2j+1}  =  \nabla f_{j+1} + \nabla e_{2j+1}.
\end{equation}

Let $\Delta g_n = g_{n+1}-g_n$,  be the forward difference operator, i.e., 
\begin{equation}
\label{eq:DC_forward_back}
\begin{cases}
  \Delta\nabla  f_n = \nabla f_{n+1} - \nabla f_n = f_{n+1} - 2 f_n + f_{n-1} \\
  \Delta\nabla  e_n = \nabla e_{n+1} - \nabla e_n= e_{n+1} - 2 e_n + e_{n-1}.
\end{cases}
\end{equation}

Then, from (\ref{eq:DC_back}),
$$
\Delta\nabla f_{2j} = \Delta\nabla f_j + \Delta\nabla e_{2j},
\quad
\Delta\nabla f_{2j+1} = \Delta\nabla e_{2j+1}
$$
for $j \geq 1$, with $\Delta\nabla f_1 = f_2 - 2f_1 = e_2 = \Delta\nabla e_1$.

Basic calculations now  show that, for any sequence $f_n,$
\begin{equation}
\label{eq:key-MSF}
f_n - nf_1 = \sum_{j<n} (n-j) \Delta\nabla f_j.
\end{equation}

Therefore, (\ref{eq:MPF-m=1}) gives that
\begin{equation}\label{eq:fn-step}
f_n - n f_1 = \frac{n}{2\pi i} \int_{c-i\infty}^{c+i\infty}  
D_f(s) \frac{n^s ds}{s(s+1)},
\end{equation}
where 
$D_f(s) := \sum_{j=1}^{\infty} \Delta\nabla f_j j^{-s}$
is the DGF of $\Delta\nabla f_j$.

Further calculation yields
\begin{eqnarray*}
&& D_f(s)\\
& = & \Delta\nabla f_1 + \sum_{j=1}^{\infty} \frac{\Delta\nabla f_{2j}}{(2j)^s} + \sum_{j=1}^{\infty} \frac{\Delta\nabla f_{2j+1}}{(2j+1)^s} \\
& = & \Delta\nabla e_1 + \left( \sum_{j=1}^{\infty} \frac{\Delta\nabla f_j}{(2j)^s} + \sum_{j=1}^{\infty} \frac{\Delta\nabla e_{2j}}{(2j)^s} \right) + \sum_{j=1}^{\infty} \frac{\Delta\nabla e_{2j+1}}{(2j+1)^s} \\
& = & \frac{1}{2^s} \sum_{j=1}^{\infty} \frac{\Delta\nabla f_j}{j^s} + \Delta\nabla e_1 + \sum_{j=1}^{\infty} \frac{\Delta\nabla e_{2j}}{(2j)^s} + \sum_{j=1}^{\infty} \frac{\Delta\nabla e_{2j+1}}{(2j+1)^s} \\
& = & \frac{D_f(s)}{2^s} + \sum_{j=1}^{\infty} \frac{\Delta\nabla e_j}{j^s}.
\end{eqnarray*}

Solving for $D_f(s)$ gives
\begin{equation}\label{eq:dgf-f-exact}
D_f(s) = \frac{1}{1-2^{-s}} \sum_{j=1}^{\infty} \frac{\Delta\nabla e_j}{j^s}.
\end{equation}

Combining (\ref{eq:fn-step}) and (\ref{eq:dgf-f-exact}) proves Lemma \ref{lem:solution-f-integral}.

\section{Proof of (\ref{eq:left=0-m=2})}\label{app:left-zero}
In this section, we mimic the proof of (\ref{eq:left=0-m=1}) in \cite{FGKPT-1994} to 
prove (\ref{eq:left=0-m=2}).

Setting $\lambda_j\equiv 1$ in
(\ref{eq:MPF-m=2}) gives
\begin{equation}\label{eq:left=0-m=2-step}
\frac{(n-1)(2n-1)}{12n} = \frac{1}{2\pi i}\int_{3-i\infty}^{3+i\infty}\zeta(s)\frac{n^s ds}{s(s+1)(s+2)}.
\end{equation}

Now consider the rectangular contour $\Gamma '$, which we defined in (\ref{eq:Gamma2}) (see Figure \ref{fig:contours_useful_identities}). By (\ref{eq:left=0-m=2-step}),
$$\lim_{R\rightarrow\infty}\frac{1}{2\pi i}\int_{\Gamma '_1}\zeta(s)\frac{n^s ds}{s(s+1)(s+2)} = \frac{(n-1)(2n-1)}{12n}.$$

By Lemma \ref{thm:bound_top_bottom_zeta_general},
$$\frac{1}{2\pi i}\int_{\Gamma '_2}\zeta(s)\frac{n^s ds}{s(s+1)(s+2)}\mbox{ and }\frac{1}{2\pi i}\int_{\Gamma '_4}\zeta(s)\frac{n^s ds}{s(s+1)(s+2)}$$
vanish as $R\rightarrow\infty$. The  poles and their residues inside the contour
can be easily computed. They are
\begin{enumerate}
\item A simple pole at $s=1$. The residue at this pole is $\frac{n}{6}$.
\item A simple pole at $s=0$. The residue at this pole is $\frac{1}{2}\zeta(0)=-\frac{1}{4}$.
\item A simple pole at $s=-1$. The residue at this pole is $-\frac{\zeta(-1)}{n} = \frac{1}{12n}$.
\end{enumerate}

By the Cauchy residue theorem, we have
\begin{eqnarray}
&& \frac{1}{2\pi i}\int_{-5/4-i\infty}^{-5/4+i\infty}\zeta(s)\frac{n^s ds}{s(s+1)(s+2)}\nonumber\\
&=& \lim_{R\rightarrow\infty}\frac{1}{2\pi i}\int_{\Gamma '_3}\zeta(s)\frac{n^s ds}{s(s+1)(s+2)}\nonumber\\
&=& \frac{n}{6} + \left(-\frac{1}{4}\right) + \frac{1}{12n} - \sum_{q=1,2,4}\lim_{R\rightarrow\infty}\frac{1}{2\pi i}\int_{\Gamma '_q}\zeta(s)\frac{n^s ds}{s(s+1)(s+2)}\nonumber\\
&=& \frac{n}{6} + \left(-\frac{1}{4}\right) + \frac{1}{12n} -\frac{(n-1)(2n-1)}{12n}\nonumber\\   
&=& 0.
\end{eqnarray}

\begin{figure}[t]
\centering%
\scalebox{0.35}{\includegraphics{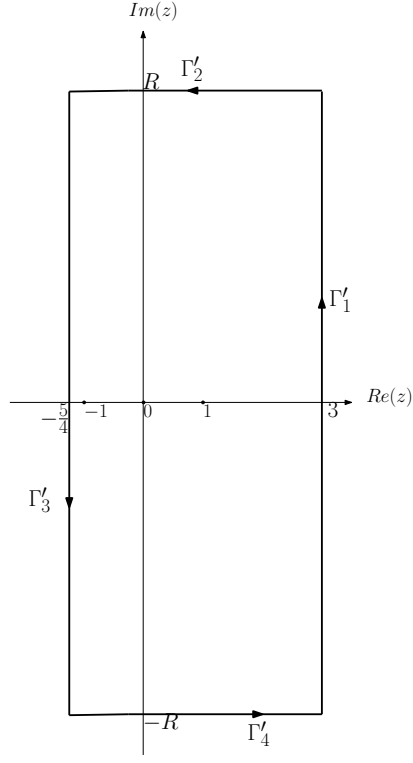}}\\
\caption{The contour $\Gamma '$ used to derive (\ref{eq:left=0-m=2}) with the poles of $\frac{\zeta(s)}{s(s+1)(s+2)}$ inside $\Gamma '$ at $-1,0,1$ noted.}
\label{fig:contours_useful_identities}
\end{figure}

\section{Proof of Lemma \ref{lem:DGF-closed}}\label{app:DGF}

Recall from Lemma \ref{lem:property-v-and-v_2} that $v(2n)=v(n)$, then
\begin{eqnarray*}
V_M(s) &=& \sum_{\mbox{\footnotesize{odd }}j}\frac{v(j)^M}{j^s} + \sum_{i=1}^\infty\frac{v(2i)^M}{(2i)^s}\\
&=& \sum_{\mbox{\footnotesize{odd }}j}\frac{v(j)^M}{j^s} + \frac{1}{2^s}\sum_{i=1}^\infty\frac{v(i)^M}{i^s}\\
&=& \sum_{\mbox{\footnotesize{odd }}j}\frac{v(j)^M}{j^s} + \frac{1}{2^s}V_M(s).
\end{eqnarray*}
This yields
$$\sum_{\mbox{\footnotesize{odd }}j}\frac{v(j)^M}{j^s} = \left(1-\frac{1}{2^s}\right)V_M(s).$$

Recall from Lemma \ref{lem:property-v-and-v_2} that $v_2(2n)=v_2(n)+1$; also, for all odd $n$, $v_2(n)=0$. We have
\begin{eqnarray*}
\sum_{j=1}^\infty\frac{v_2(j)}{j^s} &=& \sum_{i=1}^\infty\frac{v_2(2i)}{(2i)^s}\\
&=& \frac{1}{2^s}\sum_{i=1}^\infty\frac{v_2(i)+1}{i^s}\\
&=& \frac{1}{2^s}\zeta(s)+\frac{1}{2^s}\sum_{i=1}^\infty\frac{v_2(i)}{i^s}
\end{eqnarray*}
and hence
$$\sum_{j=1}^\infty\frac{v_2(j)}{j^s} = \frac{1}{2^s-1}\zeta(s).$$

Finally recalling from Lemma \ref{lem:property-v-and-v_2} that
$\nabla v(n) = v(n) -v(n-1) = 1-v_2(n)$ yields
$$\sum_{j=1}^{\infty} \frac{\nabla v(j)}{j^s} = \zeta(s) - \sum_{j=1}^{\infty} \frac{v_2(j)}{j^s} = \frac{2^s-2}{2^s-1} \zeta(s).$$

\end{document}